\def\year{2019}\relax
\documentclass[letterpaper]{article} 
\usepackage{aaai19} 
\nocopyright 

\usepackage{times}  
\usepackage{helvet}  
\usepackage{courier}  
\usepackage{url}  
\usepackage{graphicx}  
\usepackage{booktabs}       

\usepackage[linesnumbered,ruled,noend]{algorithm2e}
\usepackage{booktabs}       
\SetCommentSty{mycommfont}

\newcommand{\dict}{\textnormal{\textsf{DICT}}}
\newcommand{\lime}{\textnormal{\textsf{LIME}}}
\newcommand{\clime}{\textnormal{\textsf{CLIME}}}
\newcommand{\glime}{\textnormal{\textsf{GLIME}}}
\newcommand{\nim}{\textnormal{\textsf{NIME}}}
\newcommand{\SC}{\textnormal{\textsc{SC}}}
\newcommand{\IC}{\mathcal I}

\usepackage{xcolor}
\usepackage{gensymb} 
\usepackage{graphicx}
\usepackage{amsmath}
\usepackage{dsfont}
\usepackage{amsthm}
\usepackage{multirow} 
\usepackage{amssymb} 
\usepackage{units} 
\usepackage{mathtools}

\usepackage{chngcntr} 
\usepackage{apptools}
\AtAppendix{\counterwithin{lemma}{section}}
\AtAppendix{\counterwithin{theorem}{section}}
\AtAppendix{\counterwithin{proposition}{section}}
\AtAppendix{\counterwithin{claim}{section}}

\newcount\Comments  
\newcount\Includeappendix 
\newcount\Includeackerc 

\definecolor{darkgreen}{rgb}{0,0.6,0}
\newcommand{\kibitz}[2]{\ifnum\Comments=1{\color{#1}{#2}}\fi}

\newcommand{\omer}[1]{\kibitz{blue}{[Omer :#1]}}

\newcommand{\TODO}[1]{\kibitz{red}{[TODO:#1]}}

\usepackage{MnSymbol} 

\newcommand{\mM}{\mathcal{M}}

\newcommand{\ind}{\mathds{1}}

\newcommand{\defeq}{\stackrel{\text{def}}{=}}

\theoremstyle{plain}
\newtheorem{claim}{Claim}
\newtheorem{definition}{Definition}

\newtheorem{theorem}{Theorem}
\newtheorem{proposition}{Proposition}
\newtheorem{lemma}{Lemma}
\newtheorem{corollary}{Corollary}

\theoremstyle{definition}
\newtheorem{example}{Example}

\newcommand\bl[1]{\boldsymbol{ #1 } }

\newcommand\abs[1]{\left| #1  \right|}
\DeclareMathOperator*{\argmin}{arg\,min} 

\newcommand{\ceil}[1]{\left\lceil#1\right\rceil}
\newcommand{\floor}[1]{\left\lfloor#1\right\rfloor}

\begin{document}
%
\title{
From Recommendation Systems to Facility Location Games
}

\author{Omer Ben{-}Porat, Gregory Goren, Itay Rosenberg and Moshe Tennenholtz \\
Technion - Israel Institute of Technology\\
Haifa 32000 Israel\\
\{omerbp@campus,itayrose@campus,gregory.goren@campus,moshe@ie\}.technion.ac.il \\
}

\maketitle
\begin{abstract}
Recommendation systems are extremely popular tools for matching users and contents. However, when content providers are strategic, the basic principle of matching users to the closest content, where both users and contents are modeled as points in some semantic space, may yield low social welfare. This is due to the fact that content providers are strategic and optimize their offered content to be recommended to as many users as possible. 
Motivated by modern applications, we propose the widely studied framework of facility location games to study recommendation systems with strategic content providers. Our conceptual contribution is the introduction of a \textit{mediator} to facility location models, in the pursuit of better social welfare. We aim at designing mediators that a) induce a game with high social welfare in equilibrium, and b) intervene as little as possible. In service of the latter, we introduce the notion of {\em intervention cost}, which quantifies how much damage a mediator may cause to the social welfare when an off-equilibrium profile is adopted. As a case study in high-welfare low-intervention mediator design, we consider the one-dimensional segment as the user domain. We propose a mediator that implements the socially optimal strategy profile as the unique equilibrium profile, and show a tight bound on its intervention cost. Ultimately, we consider some extensions, and highlight open questions for the general agenda.

\end{abstract}

\section{Introduction}
\label{sec:intro}
Publishing, blogging, and content creation are fundamental to data science. Indeed, many of the AI-based recommendation systems aim at matching users with data created by content providers. While at first glance this task might be viewed as purely computational, a major topic that should be tackled is the participants' {\em incentives}.

To illustrate such incentives, consider the following example inspired by Hotelling's seminal work \cite{Hotelling}. A blogging recommendation system recommends users with relevant blogs. Two players (i.e., publishers/blog owners) write one blog each. For simplicity, assume every blog is represented as a point along the $[0,1]$ segment, e.g., the mix between news and articles in the blog. Thus, each player selects a point in that segment. A set of users, where user preferences are uniformly distributed along the same segment, approach the recommendation system to have a recommendation for one blog. The social cost of the users is defined as the sum of absolute distances between preferences and recommended contents. As Hotelling demonstrates, the only pure Nash equilibrium (PNE hereinafter) is obtained when both players select the center of the segment, producing identical content. The social cost of the PNE profile is far from optimum, and users effectively suffer from strategic behavior of the players.

The difference between this example and the one considered by Hotelling in his seminal work is that in modern applications, users are not explicitly exposed to players' content, but rather the recommendation system serves as a \textit{mediator} for matching users with content. However, despite their commercial success, current recommendation systems basically (very efficiently) implement the above principle: users are matched with the ``closest'' content. A prominent example is the vector space model for document retrieval in response to a user query \cite{salton1975vector}. Indeed, the recommendation system is just a tool for having the users ``choose'' the offered content most similar to their taste.

\subsection{The connection to facility location games}
Facility location games \cite{anderson1992discrete,brenner2010location,fournier2014hotelling}, which are extensively studied in economics, operations research and computer science, portray recommendation systems with strategic content providers incredibly well, due to the nature of recommendation systems described above. This is true since in many (and perhaps even most) recommendation systems, the computational task of matching users with content is carried out by modeling both into a joint metric space. In the vast majority of facility location games literature, it is assumed that users are attracted to their nearest facility\footnote{Some exceptions appear in Subsection \ref{subsec:related}.}. In the example above, this amounts to each user selecting the closest blog to her preference.

Given a strategy profile, i.e., the selection of each player for a particular mix of news and articles in her blog, the best (in terms of social cost) recommendation system matches each user with the nearest blog to her preference. However, the way the recommendation system operates affects the strategic behavior of players, and more precisely redefines player payoffs. As a result, it may induce a game with a different set of PNE, potentially with improved social welfare (i.e., the social cost with negative sign). This calls for the design of a mediator that takes into account strategic behavior of players, aiming at achieving low social cost in equilibrium.

\subsection{Our agenda: High welfare low intervention mediator design}
Our conceptual contribution is the introduction of a mediator to facility location models, in the pursuit of better outcomes for recommendation systems with strategic content providers. We aim at designing mediators that a) induce a game with high social welfare in equilibrium, and b) intervene as little as possible.

To better describe the latter, let \nim{} be the No Intervention Mediator, namely the default mediator that displays each user with its nearest content. A question in this regard is whether a mediator that intervenes with the process can do better than \nim{}, in terms of social welfare. Clearly, one can design a mediator that \textit{dictates} each player which content to select, thereby forcing the players to any desired strategy profile. Consider the \dict{} mediator that operates as follows: it commands each player to play a particular strategy, and if she disobeys \dict{} directs no user to her content. Indeed, by adopting \dict{} the mediator designer is guaranteed to have any specific strategy profile as the unique PNE of the induced recommendation game. However, this is a huge and potential harmful intervention, and doing so may lead to poor performance in case that specific profile is not materialized. Crucially, this can happen even when the players are rational, but some of the assumptions the mediator holds regarding the players are violated, e.g., players have constraints on their strategy space, slightly different payoff functions, they are able to form coalitions and so on.

A major aspect of our approach is the quantification of {\em intervention cost}. Given the above intuition, we define the intervention cost as the maximal increase in user social cost, which might be caused if the players use  arbitrary strategies (i.e., select arbitrary contents). Formally, let $\mM$ denote a mediator, $S$ be the set of all strategy profiles, $\bl s \in S$ be a strategy profile, and let $\SC(\mM, \bl s)$ denote the social cost under $\mM$ and $\bl s$. The intervention cost of a mediator $\mM$ is defined as $\sup_{\bl s\in S} \{ \SC(\mM, \bl s)-\SC(\nim, \bl s) \}$.
Evidently, \nim{} has a zero intervention cost and high social cost, and serves as a benchmark for low intervention cost. On the other hand, \dict{} has high intervention cost but low social cost (when dictating the socially optimal strategy profile), and it serves as a benchmark for low social cost. The main challenge in our agenda is the design of a mediator with low intervention cost \textit{and} low social cost in equilibrium.

\subsection{Our results}

This paper presents a case study in high welfare low intervention mediator design. The mathematical model we adopt in service of the above is based on pure location Hotelling games \cite{Hotelling}, with the $[0,1]$ segment. Section \ref{sec:Mathematical Formulation} presents a formal mathematical model for the setting, stating former widely-known results for pure Hotelling games (or equivalently, recommendation games with \nim{} as a mediator) on the segment with uniform user distribution. In addition, in Subsection \ref{subsec:model-ic} we define the intervention cost, and bound the intervention cost of \dict{}.

The main technical results of the paper appear in Section \ref{sec:lime}. We introduce the Limited Intervention Mediator, \lime{}. We provide the intuition behind it, as well as a thorough example to illustrate how it operates. We then prove that the game induced under \lime{} possesses a unique PNE, which attains the minimal social cost. Then, we establish upper and lower bounds on its intervention cost, which are almost tight. 
We show that not only is its intervention cost lower than that of \dict{} for every $n$, it is also $O(\frac{1}{n})$. Since \dict{} and \lime{} have the same social cost under (the unique) PNE profile, the results given in this section provide a highly encouraging answer to the challenge given above.

We subsequently study three natural extensions of the basic setting. 
First, we discuss neutral mediators. Informally, a mediator is neutral if when two players swap their strategies, the mediator's direction also swaps. We show an impossibility result for neutral mediators that aim at optimizing the social cost in the two-player case. Second, we design a mediator with a configurable intervention cost. This mediator is important where e.g. one seeks to design a mediator that minimizes social cost but is penalized for intervention (similarly to regularization in machine learning models). We show that for some cases of $n$, it induces Pareto optimal solutions for the setting. Third, we consider non-uniform user distributions. We propose the General Limited Intervention Mediator, which depends on the distribution quantiles solely. We then prove that it always induces a game with a unique PNE. This becomes even more striking when one recalls that PNE may not exist in facility location games (with no mediator) \cite{osborne1993candidate}, and are generally hard to characterize (see e.g. \cite{shilony1981hotelling,ewerhart2015mixed}). We bound its intervention cost (for the uniform distribution), and show that its intervention cost is lower than that of \dict{}. 

From an algorithmic perspective, we deal with one-dimensional problems. This may sound disappointing, but recall that we intentionally focus on relatively simple, structured problems, and that this domain is extremely well-studied in the facility location literature. Due to space constraints, the proofs are deferred to the appendix.

\subsection{Related work}
\label{subsec:related}
The notion of a mediator in a game-theoretic setting was first proposed by Aumann in his seminal work on Correlated Equilibrium \cite{aumann1974subjectivity}. In the setting Aumann considers, a mediator may send signals to the players, where the signal is designed to drive the players to achieve better payoffs. The work of \citeauthor{shoham1995social} \shortcite{shoham1995social} grants stronger capabilities to the mediator, by setting constraints on participants' behavior. A different type of mediator, considered by \citeauthor{monderer2004k} \shortcite{monderer2004k}, is allowed to change player payoffs. As in our work, the mediator is designed to ``force'' players into playing some desired subset of their strategy sets with minimal interference with the payoff functions. Unlike the above work, where the mediator intervenes with the outcome of the game to improve the social welfare of the (strategic) players, in the model we study the mediator intervenes in order to improve the social welfare of the (non-strategic) users, as we describe next.

The work of \citeauthor{basat2017game} \shortcite{basat2017game} introduces a formal model of adversarial information retrieval. Each author (a player) has several strategies, where each strategy corresponds to a document she can publish. Each pair of document-query (where a document is a player's strategy and a query represents a user) has a score, termed ``quality'', which measures the extent to which the document is relevant to the query. The mediator (i.e., search engine) provides each user the document with the highest quality relative to her query, among those selected under the corresponding strategy profile. In their model, every author seeks to maximize a function that takes into account the number of queries for which her document is the most relevant (representing users directed to that document by the mediator) along with the document-query quality.

As \citeauthor{basat2017game} show, displaying each user with the document with the highest quality (i.e., no intervention) may lead to deteriorated content and low user utility. They argue against no intervention, and claim that introducing randomization into the mediator leads to an overall user utility that transcends that of no intervention at all. Nevertheless, they neither a) provide a systematic approach for doing so under PNE profiles; nor b) show that a PNE always exists. Recently, \citeauthor{ben2018Convergence} \shortcite{ben2018Convergence} claim in favor of no intervention in a similar model, and show that no intervention policy leads to convergence of any better-response dynamics; thus, a PNE is guaranteed to exist. Importantly, \citeauthor{basat2017game} \shortcite{basat2017game}  as well as \citeauthor{ben2018Convergence} \shortcite{ben2018Convergence} do not consider the task of mediator design for better outcomes.

The work of \citeauthor{ben2018game} \shortcite{ben2018game} is the first to consider mediator design in recommendation systems with strategic content providers, in a mathematical model that extends that of \citeauthor{basat2017game} \shortcite{basat2017game} and \citeauthor{ben2018Convergence} \shortcite{ben2018Convergence}. They highlight several fairness-related properties that a mediator should arguably satisfy, along with the requirement of PNE existence.  They show that the no-intervention mediator satisfies the fairness-related properties, but may lead to a game without PNEs. Then, they design a mediator that is based on the Shapley value \cite{shapley1952value}, prove it satisfies the fairness properties and the game it induces always possesses a PNE. On the other hand,  \citeauthor{ben2018game} \shortcite{ben2018game} take an axiomatic approach, and do not address user utility as one of the axioms. Moreover, they show that the user utility under their proposed mediator can be arbitrarily low. In contrast, this paper stems from a user utility optimization point of view, and so are the solution concepts it proposes.


Several extensions of pure location Hotelling games have been suggested recently, assuming non-deterministic user behavior \cite{feldman2016variations,shen2017hotelling,ben2017shapley}. More precisely, users do not select their nearest facility, but rather have a \textit{reaction function}, mapping every strategy profile to a distribution over player indices, possibly skipping all options. The work of \citeauthor{ben2017shapley} \cite{ben2017shapley} associates users with a reaction function motivated by decision theory literature, and shows that the induced facility location game always possesses a PNE, regardless of the metric space in which users are embedded.  \citeauthor{shen2017hotelling} \shortcite{shen2017hotelling} show a result of similar flavor for a different reaction function. Interestingly, every reaction function can be implemented as a mediator. Unlike this line of research, in this paper (as in the information retrieval setting) the mediator is required to direct users to facilities w.p. 1, i.e., skipping all facilities is not an option. We elaborate on the latter in Section \ref{sec:discussion}.

Finally, a different line of research in the algorithmic game theory literature \cite{nisan2007algorithmic} is the study of facility location, in the context of approximate mechanism design without money \cite{ProcTenn09}. That literature deals with the case where only one entity dictates the place of a facility (or several facilities), while user preferences are their private information and are strategically reported (see,e.g., \cite{schummer2007mechanism}). In contrast, the setting we consider is a full information setting.

\section{Mathematical Formulation}
\label{sec:Mathematical Formulation}
\begin{table*}
  \centering
  \begin{tabular}{ccccccc}
       & \multicolumn{3}{c}{Social Cost (former results)} &  & \multicolumn{2}{c}{Intervention Cost (Sections \ref{sec:Mathematical Formulation} and 
  \ref{sec:lime})} \\
   \cmidrule(l){2-4} \cmidrule(l){6-7}
   Number of players & Optimal      & \nim{}, best PNE &\nim{}, worst PNE &  &  $\IC_n({\dict{}})$ &$\IC_n({\lime{}})$ \\
\midrule
   	    $n=2$      & $\nicefrac{1}{8}$   & $\nicefrac{1}{4}$         & $\nicefrac{1}{4}$ &  & $\nicefrac{1}{4}$ & $\nicefrac{3}{16}$     \\
   	    $n=3$      & $\nicefrac{1}{12}$  & \multicolumn{2}{c}{no PNE exists}&  &{\footnotesize$\geq 0.278$} & {\footnotesize$\in(0.236,0.278)$} \\
   	        $n\geq 4$  & $\nicefrac{1}{4n}$  & $\nicefrac{1}{4(n-2)}$    & $\nicefrac{1}{4\ceil{\nicefrac{n}{2}}}$ & & {\footnotesize $\geq \nicefrac{1}{2}-\nicefrac{3}{4n}+\nicefrac{1}{4n^2}$  }    &{\footnotesize $\in\left(\frac{2n-4}{n^2},\frac{2n-3.5}{n^2}\right)$} \\
    \bottomrule
  \end{tabular}
    \caption{A summary of former results known for pure location Hotelling games, and the results of this paper. The social cost is given for three scenarios: the optimal social cost, for the profile $\bl o^n$ (see Subsection \ref{subsec:no-intervention}); the best PNE in terms of social cost under \nim{}; and the worst PNE in terms of social cost under \nim{}. Under both \dict{} and \lime{}, the induced game possesses a unique PNE, with social cost of $\nicefrac{1}{4n}$, namely the optimal social cost. For these mediators, the table reports the bounds on the intervention cost (see Subsection \ref{subsec:model-ic}) as obtained in Sections \ref{sec:Mathematical Formulation} and 
  \ref{sec:lime}. \label{table:sc+ic}}
\end{table*}

In this section we introduce our formal mathematical model.
The non-cooperative game we consider is formally defined as follows:
\begin{enumerate}
\item A continuous density function $g$ over the unit interval $[0,1]$, representing user distribution.
\item A set of players $[n]=\{1,2,\dots,n\}$, where the pure\footnote{In this paper we discuss pure strategy profiles only.} strategy set of player $i\in[n]$ is $S_i=[0,1]$. It will sometimes be convenient to say that player $i$ has/owns a \textit{facility} in $s_i$ if the strategy she adopts is $s_i$. The set of all strategy profiles is denoted by $S\defeq \prod_{i=1}^n S_i$.
\item A mediator $\mM$ is a mapping from the set of strategy profiles and users to the set of distributions over player indices, i.e.,
$ \mM: S \times [0,1] \rightarrow \Delta([n]).$
Given a pure strategy profile $\bl s=(s_1,\dots,s_n)\in S$, a user $t\in [0,1]$ and a player $i\in [n]$, we denote by $\mM(\bl s,t)_i$ the probability that $\mM$ will send user $t$ to player $i$ under the strategy profile $\bl s$.
\item Given a pure strategy profile $\bl s\in S$, the payoff of each player $i$ is the proportion of users $\mM$ directs to her facility. Namely, $\pi_i(\bl s)= \int_{0}^1 \mM(\bl s,t)_i \cdot g(t)dt. $
Each player aims at maximizing her payoff.
\end{enumerate}
According to these assumptions, every game is fully described by the number of players, the mediator, and the user distribution function, .i.e., $G\left(n,\mM,g\right)$. In addition, unless stated otherwise, $g$ is the uniform distribution; hence, for convenience, we shall use the notation $G\left(n,\mM\right)$ to describe the $n$-player game induced by selecting $\mM$ as a mediator.

\subsubsection{Equilibrium} For a vector $\bl v = (v_1,\dots v_n)$, we denote by $\bl v_{-i}\defeq(v_1,\dots v_{i-1},v_{i+1},\dots v_n)$ the vector that does not contain the $i$-th coordinate of $\bl v$.

A strategy profile $\bl s \in S$ is called a {\em pure Nash equilibrium} (PNE hereinafter) if for every $ i \in [n]$ and every $s_i' \in S_i$ it holds that $\pi_i(s_i',\bl s_{-i}) \leq \pi_i(\bl s)$.
We say that player $i$ has a {\em beneficial deviation} under a strategy profile $\bl s=\left(s_i,\bl s_{-i}\right)$ if there exists $s_i'\in S_i$ such that $\pi_i(s_i',\bl s_{-i})  > \pi_i(\bl s)$.

\subsection{Social cost}
We remark that any mediator always directs a user to one of the facilities selected by the players; thus, the player payoffs sum to one under every strategy profile $\bl s$, i.e., $\sum_{i=1}^n \pi_i(\bl s)=1$. Consequently, we view the social cost of the users as the public welfare. In the blogging example given above, we consider the distance between a user's preferences and the actual attribute of the blog as the extent to which she is satisfied with that blog. Clearly, this perspective is identical to the one considered by \cite{Hotelling}, albeit the motivation of \cite{Hotelling} is the physical world. Motivated by transportation cost, Hotelling \cite{Hotelling} defines the cost a customer (i.e., a user) suffers as the distance he has to travel to reach his nearest facility. This notion is naturally extended to the mediated setting.
\begin{definition}[Social cost]
Given a strategy profile $\bl s$ and a mediator $\mM$, the social cost is the sum of distances users must travel to reach their recommended facility. Formally, $\SC:S \rightarrow \mathbb{R}_{\geq 0}$ is defined by
\[
{\footnotesize
\SC(\mM,\bl s) = \int_0^1 \sum_{i=1}^n\mM(\bl s,t)_i \cdot\abs{s_i-t} g(t) dt.
}
\]
\end{definition}
\subsection{No intervention mediator}
\label{subsec:no-intervention}
We denote by \nim{} the No-Intervention Mediator. Namely, \nim{} sends every user to his nearest facility, breaking ties uniformly. Formally,
\[
{\footnotesize
\nim(\bl s,t)_i =\frac{\ind_{\abs{s_i -t}\leq \min_{i'} \abs{s_{i'}-t}}}{\sum_{j=1}^n \ind_{\abs{s_j -t}\leq \min_{i'} \abs{s_{i'}-t}} }
}.
\]
By implementing \nim{} we recover the non-mediated version of the setting, where under every strategy profile each user is attracted to his nearest facility, as is in pure location Hotelling games. Given $n$, we denote by $\bl o^n=(o^n_1,\dots,o^n_n)$, where $o^n_i \defeq\frac{2i-1}{2n}$, the sequence of {\em $n$-socially optimal locations}. These are optimal in the following sense:
\begin{claim}[e.g. \cite{Ben-Porat18}]
\label{koneplayer_opt}
It holds that $ \SC(\nim, \bl o^n) = \inf_{\bl s \in S} \SC(\nim, \bl s)=\frac{1}{4n}.$
\end{claim}
It turns out that this is the unique socially optimal profile, up to renaming the players. Moreover, for any fixed strategy profile, it is apparent that employing \nim{} results in the lowest possible social cost w.r.t. that particular profile. Namely,
\begin{claim}
For every strategy profile $\bl s$ and every mediator $\mM$, it holds that $\SC(\nim,\bl s) \leq \SC(\mM,\bl s)$. 
\end{claim}
However, as we show later, the set of PNE depends on the mediator; hence, a PNE under $\nim$ may not be in equilibrium under another mediator, and vice versa. Since we care about the social cost in equilibrium, this turns out to be crucial. Pure location Hotelling games (or equivalently $G(n,\nim)$ for any $n$) have been studied extensively in the past decade, and its equilibrium structure is widely known (see, e.g. \cite{eaton1975principle,fournier2014hotelling,Ben-Porat18}). For completeness, we state the following well-known results for $G(n,\nim)$:\\
$\bullet$ The only PNE for $n=2$  is $(\nicefrac{1}{2},\nicefrac{1}{2})$;\\
$\bullet$ There is no PNE for $n=3$, although a mixed NE exists \cite{shaked1982existence};\\
$\bullet$ For $n\in\{4,5\}$ a unique PNE exists (up to renaming the players). If $n\geq 6$, there are infinite PNEs.\\
For an elaborated discussion of these results, the reader is referred to \cite{eaton1975principle}. Table \ref{table:sc+ic} summarizes the social cost for \nim{}. Notice the social cost under the worst PNE, which is a factor of two of the optimal feasible social cost.

\subsection{Intervention cost}
\label{subsec:model-ic}
Indeed, a mediator that implements the $n$-socially optimal locations as a PNE can greatly decrease the social cost. By intervening with the market and using a basic punishing mechanism, one can drive the players to play any arbitrary profile, by intervention that makes that profile the only PNE.

We denote by $\dict$ the mediator\footnote{One can think of other dictatorship mediators as well, by varying the punishment in case of disobedience.} that dictates the strategy profile $\bl o^n$, the $n$-socially optimal locations. Namely, $\dict$ operates as follows: given a strategy profile $\bl s$, let $A=\{ i: s_i=o_i^n \}$. If $A$ is not empty, $\dict(\bl s)$ directs every user $t$ to its nearest facility from those of the players in $A$, i.e., $\dict(\bl s,t)_i =\ind_{i\in A} \cdot \nim(\bl s \cap A,t)_i$. Otherwise, if $A$ is empty, it directs every user to a player selected at random. It is apparent that
\begin{claim}
Consider $G(n,\dict{})$ for any $n\geq 2$. The unique PNE is $\bl o^n$.
\end{claim}
However, sometimes the strategy profile being materialized is not the equilibrium profile. As elaborated above, a mediator intervening with the system to decrease social cost may hence have implicit negative effects. It is therefore highly desired to design a mediator that not only drives the players to a ``good'' equilibrium (in terms of social cost), but also does not intervene ``that much''. To quantify the extent to which a mediator intervenes with the natural market, we introduce the following measure.
\begin{definition}[intervention cost]
The intervention cost of a mediator $\mM$ is the maximum difference between the social cost of $\mM$ and that of \nim{} (i.e., the No-Intervention Mediator), when the maximum is taken over all strategy profiles. Formally,
\[
{\footnotesize
\IC_n(\mM)\defeq\sup_{\bl s=(s_1,\dots s_n) \in S} \left\{ \SC(\mM,\bl s)-\SC(\nim,\bl s) \right\},
}
\]
where the subscript $n$ emphasizes the dependency on the number of players.
\end{definition}
Indeed, the intervention cost captures the measure of intervention by the mediator. By definition of \nim{}, it has an intervention cost of zero, i.e., $\IC_n(\nim)=0$ for every $n$. In addition, the intervention cost as defined above also demonstrates the great amount of intervention employed by \dict{}. 

To illustrate it, consider the consequences of a strategy profile other than $\bl o^n$ being materialized under \dict{}. Take for example the two-player game, $G(2,\dict{})$, and the strategy profile $\bl s =(s_1,s_2)=(0,1)$. This profile represents the case where the offered contents are varied to the extreme. Both players defy their commands; thus, \dict{} sends every user to a facility selected uniformly at random, and every user travels a distance of $\nicefrac{1}{2}$ in expectation. Summing over all users we get, $\SC(\dict,\bl s)=\nicefrac{1}{2}$. The social cost under no mediation is $\SC(\nim,\bl s)=\nicefrac{1}{4}$; hence, the intervention cost of \dict{} is lower bounded by $\SC(\dict s)-\SC(\nim,\bl s) =\nicefrac{1}{4}$. As we show in the appendix, for $n=2$ this is not only a lower bound but rather the actual intervention cost of \dict{}.

In fact, this profile can be extended to any $n$-player game, showing that $\IC_n(\dict)\geq \nicefrac{1}{4}$ for any $n$.  However, using a different construction we show a much tighter bound that depends on $n$. More precisely,
\begin{lemma}
\label{lemma:dict high intervention cost}
For any $n\geq 3$, it holds that $\IC_n(\dict) \geq \frac{1}{2}-\frac{3}{4n}+\frac{1}{4n^2}$.
\end{lemma}
Due to its high intervention cost (among other properties) using \dict{} as a mediator may not be the best solution. In the next section we devise a mediator that implements $\bl o^n$ as the unique  equilibrium, but unlike $\dict$ has a substantially low intervention cost.

\section{Limited Intervention Mediator (LIME)}
\label{sec:lime}

\begin{figure*}[t]
\centering
\includegraphics[trim={0 0cm 0 0cm},scale=1.0]{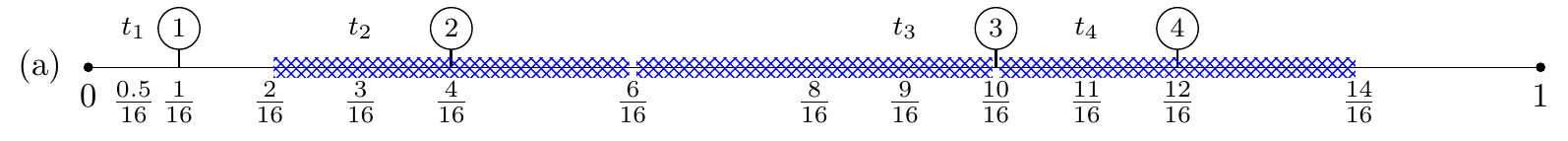}
\includegraphics[scale=1.0]{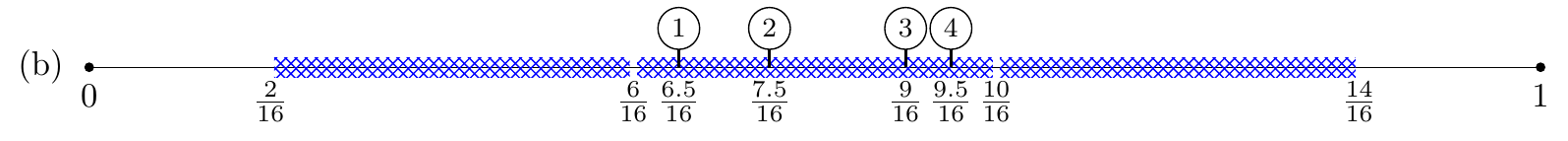}

\includegraphics[trim={0 0cm 0 0cm},scale=1.0]{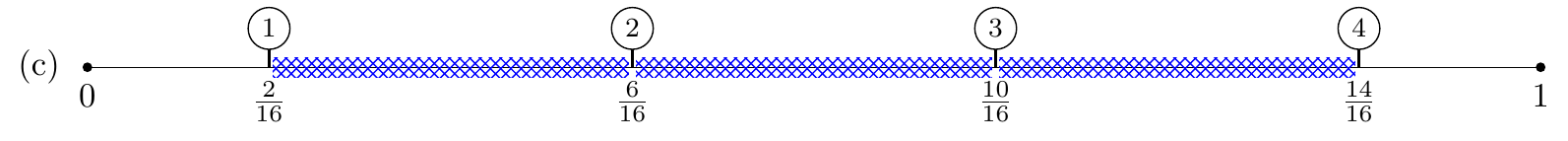}

\caption{Three strategy profiles for $G(4,\lime)$.\omer{general - the 6 should be 16!}
The blue (dotted) segments represent potentially intervened intervals (PIIs), and the numbered circle are the locations selected by the corresponding players. Sub-figure (a) visualizes Example \ref{example:lime}. \omer{I changed this one}Sub-figure (b) exemplifies a case where all players locate their facilities in the PII $(\nicefrac{6}{16},\nicefrac{10}{16})$. In this case, \lime{} directs every customer outside the PIIs to his nearest facility; every customer in $(\nicefrac{6}{16},\nicefrac{10}{6})$ is also directed to her nearest facility (Line \ref{alg:lime all in one} of Algorithm \ref{alg:lime}); and a customer inside $(\nicefrac{2}{16},\nicefrac{6}{16})$ or $(\nicefrac{10}{16},\nicefrac{14}{16})$ is directed to his nearest facility w.p. $1-\epsilon$, and with the remaining probability to a facility selected uniformly at random (Line \ref{alg:lime l or r}). The unique PNE of this game (due to Theorem \ref{thm: lime unique pne}) is demonstrated in Sub-figure (c). \label{fig:lime-example-1}}
\end{figure*}

\begin{algorithm}[t]
\caption{Limited Intervention Mediator \label{alg:lime}}
\DontPrintSemicolon
 \KwIn{ A strategy profile $\bl s$ and a user $t$}
 \KwOut{A location in $\bl s$}
  \uIf{ $\exists i\in [n-1]$ such that $t \in \left(\frac{2i-1}{2n},\frac{2i+1}{2n}  \right)$}{\vspace{0.1em} %
    $\bl l \gets \bl s \cap \left[0,\frac{2i-1}{2n}  \right], \bl r \gets \bl s \cap \left[\frac{2i+1}{2n},1  \right]$ \;
		\uIf{$\bl l \neq \emptyset$ and $\bl r \neq \emptyset$ \label{alg:lime two sides if}}
		{\KwRet $\nim(\bl l \cup \bl r, t)$\label{alg:lime l and r}  }
		\uElseIf{
		$\bl l \neq \emptyset$ or $\bl r \neq \emptyset$
		}    
		{
		 w.p. $1-\epsilon$ \KwRet $\nim(\bl l \cup \bl r, t)$, otherwise \KwRet $random(\bl s,t)$ \label{alg:lime l or r}
		}
		\uElse(\tcp*[f]{all facilities are inside $\left(\frac{2i-1}{2n},\frac{2i+1}{2n}  \right)$}){
		\KwRet $\nim(\bl s, t) \label{alg:lime all in one}$
		}
  }
  \uElse(\tcp*[f]{ $t$ is outside $\cup_{i\in[n-1]} \left(\frac{2i-1}{2n},\frac{2i+1}{2n}  \right)$}){
    \KwRet $\nim(\bl s, t)$ \label{alg:lime outside} \; 
  }\label{lime:outside}
  \SetKwProg{Fn}{Function}{:}{}
    \SetKwFunction{FMain}{$random(\bl s, t)$}
    \Fn{\FMain}{
		\KwRet an element from $\bl s$ u.a.r.
        \;
  }   
\end{algorithm}
In this section we take advantage of the game structure to devise a mediator that encourages the players to select the $n$-socially optimal locations $\bl o^n$, but intervenes substantially less than \dict{}. This is done by exploiting the equilibria structure under no mediation. According to the characterization of PNE profiles given in \cite{eaton1975principle}, under \nim{} a profile can be in equilibrium only if its peripheral facilities (i.e., leftmost and rightmost facilities) are paired. In addition, if beneficial deviations do not exist, the proportion of users coming from the left/right of any facility cannot be greater than the total proportion of users served by any other facility. 


Consider the Limited Intervention Mediator described in Algorithm \ref{alg:lime}, and referred to as \lime{} hereinafter for abbreviation. The intuition behind \lime{} is the following: between every two locations that belong to $\bl o^n$, we construct a \textit{potentially intervened interval} (PII). The users in every PII are not directed to a facility located in the same PII, but rather to the closest facility outside of it. In addition, if a PII does not have a facility located from its left or from its right (exclusive or), the users in that PII are sent to a random player w.p. $\epsilon$, where $\epsilon>0$ is an arbitrarily small constant.

To facilitate understanding of the \lime{} mediator, we provide the following example, which is further illustrated in Figure \ref{fig:lime-example-1}.a.
\begin{example}
\label{example:lime}
Consider $G(4,\lime{})$, and the strategy profile $\bl s=\left(\nicefrac{1}{16},\nicefrac{4}{16},\nicefrac{10}{16},\nicefrac{12}{16}\right)$.
User $t_1$, located in $\nicefrac{0.5}{16}$, is directed by \lime{} to player 1's facility in $\nicefrac{1}{16}$, since $t_1$ is outside the PIIs (Line \ref{alg:lime outside}). User $t_2$, who is located in $\nicefrac{3}{16}$, is inside the PII $\left(\nicefrac{2}{16},\nicefrac{6}{16}\right)$. Notice that this PII has a facility from its left and a facility from the right. Namely, according to the if condition in Line \ref{alg:lime two sides if}, $\bl r = \bl s \cap \left[\nicefrac{6}{16},1\right]\neq \emptyset$, and $\bl l = \bl s \cap \left[0,\nicefrac{2}{16}\right]\neq \emptyset$; thus, $t_2$ is directed to his nearest facility outside $\left(\nicefrac{2}{16},\nicefrac{6}{16}\right)$, which is $\nicefrac{1}{16}$ (Line \ref{alg:lime l and r}). User $t_3$ is inside the PII $\left(\nicefrac{6}{16},\nicefrac{10}{16}\right)$, and is therefore directed to his nearest facility outside this PII, player 3's facility in $\nicefrac{10}{16}$ (PIIs are open intervals, and this facility lies in the exterior of $\left(\nicefrac{6}{16},\nicefrac{10}{16}\right)$). User $t_4$ is inside $\left(\nicefrac{10}{16},\nicefrac{14}{16}\right)$, and this PII does not have a facility from its right (i.e., $\bl r=\bl s \cap \left[\nicefrac{14}{16},1\right] = \emptyset$); thus, w.p. $1-\epsilon$ \lime{} directs him to the facility in $\nicefrac{10}{16}$, and with the remaining probability he is directed to a facility selected uniformly at random (Line \ref{alg:lime l or r}).
\end{example}
 The only event not covered by Example \ref{example:lime} is the case where all the facilities are located in the same PII, as illustrated in Figure \ref{fig:lime-example-1}.b.\omer{I changed this one}

\subsection{Pure Nash equilibrium}
We now show that \lime{} is carefully constructed to mitigate players' incentives. In particular, we show that $\bl o^n$ is the unique PNE of $G(n,\lime)$. 
First, we show that under any PNE, players are not encouraged to locate their facilities in PIIs.
\begin{lemma}
\label{notininterval}
Consider $G(n,\lime{})$ for any $n\geq 2$. If $\bl s$ is an equilibrium profile, then $s_j \notin\left(\frac{2i-1}{2n},\frac{2i+1}{2n}  \right)$ for every $j\in [n]$ and $i\in [n-1]$.
\end{lemma}
Next, we leverage Lemma \ref{notininterval} to show that there is a unique equilibrium under \lime{}, which is composed of the $n$-socially optimal locations, i.e., $\bl o^n$. The case of $n=2$ is an interesting exception and is therefore discussed separately in Subsection \ref{subsec:neutral mediators}.
\begin{theorem}
\label{thm: lime unique pne}
Consider $G(n,\lime{})$ for any $n\geq 3$.
The unique PNE (up to renaming the players) is $\bl o^n$.
\end{theorem}
See Figure \ref{fig:lime-example-1}.c for illustration. Importantly, Theorem \ref{thm: lime unique pne} holds for any $\epsilon<\nicefrac{1}{2}$, and implies that $\bl o^n$ is an \textit{exact} PNE of $G(n,\lime{})$ and not an approximated PNE. In addition, under the profile  $\bl o^n$ all the facilities are outside the PIIs; hence, every user is directed to his nearest facility. Consequently,
\begin{corollary}
Consider $G(n,\lime{})$ for any $n\geq 3$. The unique PNE, $\bl o^n$, attains the optimal social cost, $\SC(\lime,\bl o^n)=\frac{1}{4n}$.
\end{corollary}
Therefore, the social cost of \lime{} under the PNE profile matches the optimal (see Table \ref{table:sc+ic}).

\subsection{Intervention cost}
\label{subsec:lime-ic}
Having demonstrated that \lime{} obtains the optimal social cost in the (unique) equilibrium, we now claim that its intervention cost is low. By definition of the \lime{} mediator, it suffices to consider users inside PIIs only, as users outside PIIs are directed to their nearest facility; hence, they do not contribute to the intervention cost. First, we lower bound its intervention cost. \omer{write that we consider epsilon to be small, and neglect it in the analysis}
\begin{lemma}
\label{lemma:lime lower bound}
For any $n\geq 3$, it holds that $\IC_n(\lime)\geq \frac{2n-4}{n^2}$.
\end{lemma}
More precisely, the bound is \omer{I put epsilon over 2}$(1-\frac{\epsilon}{2})\left(\frac{2n-4}{n^2}\right)$, which increases as $\epsilon$ decreases. The expression given in Lemma \ref{lemma:lime lower bound} is obtained by taking $\epsilon$ to zero. 
Next, we move on to upper bounding its intervention cost. 
\begin{theorem}
\label{thm:LIME-IC}
For any $n\geq 3$, it holds that $\IC_n(\lime) \leq  \frac{2n-3.5}{n^2}$. 
\end{theorem}
The proof of Theorem \ref{thm:LIME-IC} assumes for simplicity that $\epsilon$ (in Line \ref{alg:lime l or r} of Algorithm \ref{alg:lime}) is arbitrarily small.
Notice that the upper bound almost matches the lower bound. We summarize the intervention cost of \lime{} in Table \ref{table:sc+ic}. Observe that not only does \lime{} intervene less than \dict{}, but also its intervention cost diminishes in a $\nicefrac{1}{n}$ scale, similarly to the optimal social cost (which is $\nicefrac{1}{4n}$).

\section{Extensions}
\label{sec:extensions}
\omer{add $n=2$ for any dist PNE with improved SC, greg's project from back then}
Beyond the main analysis of the paper given in the previous sections, we find it important to examine the three following extensions.

\subsection{Neutral mediators}
\label{subsec:neutral mediators}
Dilemmas other than intervention cost may arise when implementing a mediator. To illustrate such dilemmas and intertwine our recommendation games with social choice theory, we analyze the two-player game. As Theorem \ref{thm: lime unique pne} shows, $\bl o^n$ is the unique PNE for $n\geq 3$. In $G(2,\lime)$, however, \lime{} does not induce a game with a unique PNE. 
\begin{proposition}
\label{prop:two-player eq}
Consider $G(2,\lime)$. A strategy profile $(s_1,s_2)$ is a PNE if and only if $(s_1,s_2)\in \left\{\frac{1}{4},\frac{3}{4}\right\}\times \left\{\frac{1}{4},\frac{3}{4}\right\}$.
\end{proposition}
Notice that the selection of the socially optimal locations remains a PNE, similarly to greater values of $n$, but this PNE is not \textit{unique}. To clarify the intuition behind the latter phenomena, consider the profile $(\nicefrac{1}{4},\nicefrac{1}{4})$. According to Line \ref{alg:lime l or r} in Algorithm \ref{alg:lime}, \lime{} directs all users in the segment $(\nicefrac{1}{4},\nicefrac{3}{4})$ to a random player w.p. $\epsilon$. Under this profile, random direction collides with directing these users to their nearest facility; hence, the randomness \lime{} employs for breaking symmetry is not enough and other tools should be applied. In fact, this problem can be easily solved by adding a lexicographic tie breaker in case both players select the same location. When augmented to \lime{}, this tie-breaker induces a game with the socially optimal locations as the unique PNE. 

Indeed, our notion of intervention cost is not affected by such tie-breaking intervention. However, the mediator should treat players \textit{fairly}, in some sense of the word. This fact is meaningful with respect to the generality of our agenda, since it introduces other considerations that involve the players rather than the users, borrowing ideas from social choice theory. The most immediate characterization of fairness in our setting is through the notion of \textit{neutrality}. Formally, a mediator $\mM$ is neutral if for every $i,j \in [n]$ and $\bl s= (\bl s_{-\{i,j\}},s_i,s_j)$ it holds that $\pi_i(\bl s)=\pi_j(\bl s_{-\{i,j\}},s_j,s_i)$. Noticeably, the \dict{} mediator is not neutral, while the \lime{} mediator and all other mediators discussed in the paper and the appendix are neutral.
The inability of \lime{} to handle the two-player game then becomes apparent: it turns out that for every neutral mediator that induces a game with a PNE, there is a PNE when both players select the same location (see the proof in the appendix). So, no neutral mediator implementing the socially  optimal locations as a unique PNE exists for $n=2$, while for $n \geq 3$ we show \lime{} does precisely that under low intervention cost.



\subsection{Configurable intervention cost}
In some situations it might be beneficial to control the amount of intervention. A mediator may offer a tradeoff between social cost in equilibrium and intervention cost. Namely, we devise a parameterized mediator, where its parameter determines a desired intervention cost. The mediator attains the desired intervention cost but suffers some increase in its social cost in equilibrium (compared to the optimal one). For some values of $n$ we show decisive positive results for the applicability  of such a mediator in creating desired tradeoffs. This extension is deferred to the appendix.

\subsection{Non-uniform user distribution}
\begin{algorithm}[t]
 \caption{\glime\label{alg:glime}}
\DontPrintSemicolon
 \KwIn{ A strategy profile $\bl s$, a user $t$}
 \KwOut{A location in $\bl s$}
  $a_i \gets q_{\nicefrac{2i-1}{2n}}$ for $i\in [n-1]$ \;
  \uIf{ $\exists i\in [n]$ such that $t \in (a_i,a_{i+1})$}{
    $\bl l \gets \bl s \cap \left[0,a_i  \right], \bl r \gets \bl s \cap \left[a_{i+1},1  \right]$ \;
		\uIf{$\bl l \neq \emptyset$ and $\bl r \neq \emptyset$}
		{w.p. $\frac{1}{2}$ \KwRet $\nim(\bl l, t)$, otherwise \KwRet $\nim(\bl r, t)$ \label{alg:glime:half half}}
		\uElseIf{
		$\bl l \neq \emptyset$ or $\bl r \neq \emptyset$
		}    
		{
        w.p. $1-\epsilon$ \KwRet $\nim(\bl l \cup \bl r, t)$, otherwise \KwRet $random(\bl s,t)$ \label{alg:clime l or r}
		}
		\uElse{
		\KwRet $\nim(\bl s, t)$
		}
  }
  \uElse{
    \KwRet $\nim(\bl s, t)$ \; 
  }  
\end{algorithm}
The results obtained for \lime{} do not hold for general user distribution, since players may have beneficial deviations. However, the case of non-uniform user distribution is appealing, especially given the fact that PNEs are hard to characterize (see, e.g., \cite{shilony1981hotelling,ewerhart2015mixed}) or may not even exist \cite{osborne1993candidate}. In this subsection we devise a mediator for general user distribution, showing initial results for this setting as well. Let $g$ be an arbitrary continuous function supported in the $[0,1]$ segment, and denote by $G(n,\mM,g)$ the game induced by the mediator $\mM$, $n$ players and $g$ as the user distribution function. Let $q_a$ be the $a$-th quantile of $g$. Consider the General Limited Intervention Mediator described in Algorithm \ref{alg:glime}, and referred to as \glime{} hereinafter for abbreviation. \glime{} operates similarly to \lime{}, but carries out two things differently. First, it incorporates the quantiles of $g$ instead of those of the uniform distribution ($\bl o^n$ defined earlier). The other difference is given in Line \ref{alg:glime:half half}. Instead of directing users in PIIs to their nearest facility outside that PII, it directs such users w.p. $\nicefrac{1}{2}$ to the nearest facility right of the PII (if such exists), and with the remaining probability to the nearest facility left of the PII. Importantly, \glime{} induces games with a unique PNE, which is highly desired. 
\begin{lemma}
\label{lemma:glime pne}
Consider $G(n,\lime{},g)$ for any $n\geq 3$ and general density function $g$. The unique PNE (up to renaming the players) is $\bl s = (q_{\nicefrac{2i-1}{2n}})_{i=1}^{n}$.
\end{lemma}
Next, we bound the intervention cost of \glime{}. Since the analysis for general distributions is highly challenging, we leave it for future work and only focus on its intervention cost under the uniform distribution.
\begin{proposition}
\label{lemma:glime ic lower}
For any $n\geq 2$, it holds that $\IC_n(\glime)\geq \frac{1}{4}-\frac{1}{2n}+\frac{1}{2n^2}$. 
\end{proposition}
Indeed, while the intervention cost of \glime{} is greater than that of \lime, it is better than that of \dict.
\begin{lemma}
\label{lemma:glime ic upper}
For every $n$, $\IC_n(\dict) > \IC_n(\glime)$.
\end{lemma}

\section{Discussion}
\label{sec:discussion} 
Another interesting solution concept that can be employed, if the domain permits it, is to allow the mediator to direct users in some cases with w.p. less than 1. Namely, to skip the players offers and avoid producing a recommendation. \citeauthor{ben2018game} \shortcite{ben2018game} implement this idea while adopting an axiomatic approach to player fairness, and introduces a mediator that induces a PNE for games with any user space (e.g., $k$-dimensional space or even non-metric space). However, they do not discuss user welfare. This suggests one should seek mediators that balance user welfare and player welfare simultaneously, hopefully by not intervening too much. 

{\ifnum\Includeackerc=1{
\section*{Acknowledgments}\label{sec:Acknowledgments}
This project has received funding from the European Research Council (ERC) under the European Union's Horizon 2020 research and innovation programme (grant agreement n$\degree$  740435).}\fi}


{\ifnum\Includeappendix=1{ 
\appendix
\onecolumn

\section{Omitted proofs from Section \ref{sec:Mathematical Formulation}}

\begin{proof}[Proof of Lemma \ref{lemma:dict high intervention cost}]
Let $\delta > 0$ be an arbitrary small constant, and let $\bl s^{\delta}=(s_1^\delta,s_2^\delta,\dots, s_n^\delta)$ denote a strategy profile such that $s^\delta_1 = \frac{1}{2n}$ and $s_i^{\delta} = \frac{2i-1}{2n}+\delta$ for $i\in\{2,\dots,n\}$.  Since $s_1=o^n_1$, all users are directed to $s_1$; hence,
\[
\SC(\dict,\bl s^\delta) = \frac{1}{2}\left( \frac{1}{2n} \right)^2+\frac{1}{2}\left( 1-\frac{1}{2n} \right)^2=\frac{1}{2}-\frac{1}{2n}+\frac{1}{4n^2}.
\]
In addition, due to the continuity of the social cost, $\SC(\nim,\bl s^\delta)=\SC(\nim,\bl o^n)+O(\delta)$ where $O(\delta)$ approaches 0 as $\delta $ decreases; Therefore,
\[
\SC(\dict,\bl s^\delta) - \SC(\nim,\bl s^\delta)=\frac{1}{2}-\frac{3}{4n}+\frac{1}{4n^2}-O(\delta).
\]
Finally, 
\begin{align*}
\IC_n(\dict)&=\sup_{\bl s\in S}  \left\{  \SC(\dict,\bl s) - \SC(\nim,\bl s) \right\}\\
&\geq \lim_{\delta \rightarrow o^+}
\left\{ \SC(\dict,\bl s^\delta) - \SC(\nim,\bl s^\delta) \right\}
\\
&=\lim_{\delta \rightarrow 0^+}\left(\frac{1}{2}-\frac{3}{4n}+\frac{1}{4n^2}-O(\delta)\right)=\frac{1}{2}-\frac{3}{4n}+\frac{1}{4n^2}.
\end{align*}
\end{proof}

\section{Omitted proofs from Section \ref{sec:lime}}

\begin{proof}[Proof of Lemma \ref{notininterval}]
Let $\bl s$ be an arbitrary profile, and assume $s_j \in\left(o^n_i,o^n_{i+1}  \right)$ holds for some $j$ and $i$. For simplicity, we assume that $\epsilon < \frac{1}{3}$. In addition, for every strategy profile $\bl s$, let $\epsilon(\bl s)$ denote the expected amount of users directed according to Line \ref{alg:lime l or r} of \lime{} (i.e. sent w.p. $\epsilon$ to a random player), and notice that $\epsilon(\bl s) < \epsilon$.

We prove the lemma by showing that player $j$ must have a beneficial deviation. We proceed by an exhaustive case analysis:
\begin{itemize}
\item If both $o^n_i$ and $o^n_{i+1}$ are occupied under $\bl s$: player $j$ gets only randomly-sent users; thus, by moving to $o^n_{i+1}$ her payoff strictly increases.
\item If exactly one of $o^n_i,o^n_{i+1}$ is occupied under $\bl s$ (w.l.o.g. $o^n_i$): we have two sub-cases. If 
$\bl s \cap \left(o^n_{i+1} ,1 \right] \neq  \emptyset$, the player controlling the rightmost facility in  $\left( o^n_i,o^n_{i+1} \right)$ can improve her payoff if she deviates to $o^n_{i+1}$.

Else, if $\bl s \cap \left(o^n_{i+1} ,1 \right] = \emptyset$, let w.l.o.g. $s_j$ be the rightmost location under $\bl s$,
\[
\pi_j(\bl s) = \frac{1}{2n}+\frac{n-i-1}{n}(1-\epsilon)+\frac{\epsilon(\bl s)}{n}.
\]
Let $s_j'=o^n_{i+1}$. It holds that
\[
\pi_j(\bl s_{-j},s_j') =  \frac{1}{2n}+\frac{n-i-1}{n}(1-\epsilon)+\frac{1}{2n}+\frac{\epsilon(\bl s_{-j},s_j')}{n} >\pi_j(\bl s).
\]

\item  If both of $o^n_i,o^n_{i+1}$ are empty under $\bl s$: we have three sub-cases.
\begin{enumerate}
\item If $s_j$ has both left and right neighbors outside $\left(o^n_i,o^n_{i+1}  \right)$: Let $j'$ denote a player with the leftmost facility in $\left(o^n_i,o^n_{i+1}  \right)$. We now show that $j'$ has a beneficial deviation. If the facility of player $j'$ is lone, then she can increase her payoff by deviating to $o^n_i$. Otherwise, if player $j's$ facility is paired, let $\alpha$ and $\beta$ denote the user mass served by the facilities in $s_{j'}$ from left and right correspondingly. If the facility of player $j'$ has a right neighbor inside $\left(o^n_i,o^n_{i+1}  \right)$, then $\beta=0\neq \alpha$, and she could deviate to $o^n_i$ and increase her payoff. Alternatively, if her facility does not have a right neighbor inside $\left(o^n_i,o^n_{i+1}  \right)$, then $\pi_{j'}(\bl s) = \frac{\alpha+\beta}{2}$, and by deviating to one  of $\left\{o^n_i,o^n_{i+1}  \right\}$ she could get at least
\[
\max\left\{\alpha+\frac{1}{2n}  ,\beta+\frac{1}{2n}\right\} > \pi_{j'}(\bl s).
\]
\item If $s_j$ has only left (symmetric for right) neighbor outside $\left(o^n_i,o^n_{i+1}  \right)$: Let $j'$ denote a player with the leftmost facility in $\left(o^n_i,o^n_{i+1}  \right)$. If her facility has a right neighbor inside $\left(o^n_i,o^n_{i+1}  \right)$, then she can increase her payoff by deviating to $o^n_i$. Otherwise, it does not have a right neighbor inside that interval. If her facility is lone, she can again increase her payoff by deviating to $o^n_i$. If her facility is paired, let $\alpha$ and $\beta$ denote the user mass served by the facilities in $s_{j'}$ from left and right correspondingly, and let $\gamma$ denote the amount of randomly-directed users she obtains. Notice that $\pi_{j'}(\bl s) = \frac{\alpha+\beta}{2}+\gamma$. However, by deviating to one of $\left\{o^n_i,o^n_{i+1}  \right\}$ she could get at least
\[
\max\left\{
\alpha +\frac{1}{n}+\gamma 
,\beta +\frac{1}{2n}+\gamma - \frac{\epsilon}{n^2}
\right\} >\pi_{j'}(\bl s).
\]
\item Else, all the facilities are inside $\left(o^n_i,o^n_{i+1}  \right)$. If all the facilities are located in the same point, the payoff of each player is $\frac{1}{n}$. By deviates to one  of $\left\{o^n_i,o^n_{i+1}  \right\}$, player $j$  can get $\max\{\frac{2i-1}{2n},\frac{2i+1}{2n}\}$, which is greater than $\frac{1}{n}$ for every $i\in[n-1]$. Otherwise, if the facilities are scattered over two points or more inside the interval, let $j'$ denote the player with the rightmost facility. Clearly, she can increase her payoff by deviating to $o^n_{i+1}$.
\end{enumerate}
\end{itemize}
We conclude that if $\bl s$ in a PNE profile, then all PIIs are free of facilities.
\end{proof}

\begin{proof}[Proof of Theorem \ref{thm: lime unique pne}]
Let $\bl s=(s_1,\dots s_n)$ be a profile such that $s_1 \leq s_2 \leq \cdots \leq s_n$, let $o^n_i = \frac{2i-1}{2n}$ for $i\in [n]$ and assume $\bl s \neq (o^n_1,\dots ,o^n_n)= \bl o^n$ and that $\bl s$ is in equilibrium. The proof consists of several steps, where in each step we further characterize $\bl s$, relying on the fact that it is a PNE. Finally, we show that $\bl s$ must be $\bl o^n$.

\noindent \textit{Step 1:} Under $\bl s$, every player gets exactly $\frac{1}{n}$. 
Otherwise, there exists a player $j$ with payoff less than $\frac{1}{n}$. Notice that $\bl o^n \setminus \bl s \neq \emptyset$ since $\bl s \neq \bl o^n$; thus, there exists a free location in $\bl o^n$, and player $j$ has a beneficial deviation: if she deviates to that location her payoff becomes greater or equal to $\frac{1}{n}$.

\noindent \textit{Step 2:} It holds that $\bl s \cap  \cup_{i=1}^{n-1} \left(o^n_i,o^n_{i+1}  \right) =\emptyset$. This follows immediately from Lemma \ref{notininterval}.

\noindent \textit{Step 3:} It holds that $\bl s \cap  \left[0,o^n_1  \right) =\emptyset$ and $\bl s \cap  \left(o^n_{n} ,1 \right] =\emptyset$.
Otherwise, assume by contradiction that $\bl s \cap  \left[0,o^n_1  \right) \neq \emptyset$ (symmetric for the other argument), and denote by $j$ a player such that $s_j \in  \left[0,o^n_1  \right)$. If player $j$'s facility is paired (namely, there is another facility of different player in $s_j$), these two facilities must obtain the same user mass from both sides, or otherwise she has a beneficial deviation. However, if this is true, her payoff is  $s_j<\frac{1}{2n}$ (with, perhaps, an additional $\frac{epsilon}{n}$), which contradicts Step 1. If player $j$'s facility is lone, we have two sub-cases. If $\bl s_{-j} \cap \left[0,o^n_1 \right]=\emptyset$ (i.e. no other facility is located in the interval $\left[0,o^n_1 \right]$),  player $j$ can increase her payoff by deviating to $o^n_1$. Else, if $\bl s_{-j} \cap \left[0,o^n_1 \right]\neq \emptyset$, then there is a facility in the segment $\left[0,o^n_1\right]$ that attracts a user mass of at most $\frac{1}{2n}$, which contradicts Step 1.

\noindent\textit{Step 4:}
We have learned so far that under $\bl s$ all the facilities are located in a subset of the socially optimal locations $\bl o^n$. In this step we show that if there exists $i\in[n-1]$ such that $o^n_i$ is not occupied, then $o^n_{i+1}$ must be occupied with more than one facility. First, assume that there exists $i\in[n-1]$ such that both $o^n_i$ and $o^n_{i+1}$ are not occupied. In this case, any player that moves to $o^n_i$ obtains a payoff that is strictly greater than $\frac{1}{n}$. 
 Second, if there exists $i\in[n-1]$ such that $o^n_i$ is not occupied, and $o^n_{i+1}$ is occupied with a lone facility, then the corresponding player obtains a payoff that is strictly greater than $\frac{1}{n}$. The mirror argument, which holds symmetrically, implies that for every $i\in \{2,\dots n\}$, if $o^n_{i+1}$ is not occupied then $o^n_i$ is occupied with more than one facility. Together, it means that for every $i\in \{2...n-1\}$, if $o^n_i$ is not occupied, then $o^n_{i-1}$ and $o^n_{i+1}$ must be occupied  with multiple facilities.

\noindent\textit{Step 5:} Both $o^n_1$ and $o^n_{n}$ are occupied. Otherwise, assume that $o^n_1$ is not occupied (symmetric for $o^n_{n}$). Since $o^n_{2}$ must be occupied with $k\geq 2$ facilities (Step 4), we have two options: 
\begin{itemize}
\item If $ \bl o^n \setminus \{ o^n_1 \}  =\bl s$, i.e. the only empty location in $\bl o^n$ under $\bl s$ is $o^n_1$, then there exists a player that gets strictly greater than $\frac{1}{n}$ (her share plus $\frac{\epsilon}{n}$). This in the only place in the proof that requires $n\geq 3$.
\item Else, if there is $i\in\{3,\dots,n\}$ such that $o^n_i \notin \bl s$, one of the players in $o^n_{2}$ can move her facility to $o^n_i$ and get a payoff that is strictly greater than $\frac{1}{n}$ (her share plus $\frac{\epsilon}{n}$). 
\end{itemize}

\noindent\textit{Step 6:} To obtain a contradiction, we must show that each location of $\bl o^n$ is occupied with a facility under $\bl s$. Let $\alpha_{mul}$ denote the number of locations in $\bl o^n$ with multiple facilities under $\bl s$. In addition, let $\alpha_{empty}$ denote the number of locations in $\bl o^n$ which are not occupied with facilities under $\bl s$. By the way we defined $\alpha_{mul}$ and $\alpha_{empty}$, it holds that $\alpha_{mul}  \leq \alpha_{empty}$.

On the other hand, Step 5 suggests that $o^n_1$ and $o^n_{n}$ must be occupied, and Step 4 implies that if $o^n_i$ is not occupied for $i\in \{2...n-1\}$, then $o^n_{i-1}$ and $o^n_{i+1}$ must be occupied with multiple facilities. As a result, $\alpha_{mul} \geq \alpha_{empty}+1$, which is clearly a contradiction.

We obtained $\bl s = \bl o^n$, as required.
\end{proof}

\begin{proof}[Proof of Lemma \ref{lemma:lime lower bound}]
Let $n\geq 3$, and let $0<\delta<\frac{1}{2n}$ be an arbitrary constant. denote $\bl s^\delta =( s^\delta_1,\dots  s^\delta_n)$ such that for every player index $i \leq \ceil{ \frac{n}{2}}$, $s_i^{\delta}=\frac{1}{2n}+\delta$, and $ s^\delta_i =\frac{2n-1}{2n}-\delta$ otherwise if $\ceil{\frac{n}{2}}<i\leq n$. To ease notation, let $\SC_A(\lime,\bl s^{\delta})$ denote the social cost when restricted to the users in a $A\subset [0,1]$ only, and similarly for $\SC_A(\nim,\bl s^{\delta})$. Due to the \lime{} operates, it follows that
\[
\SC(\lime,\bl s^{\delta}) - \SC(\nim,\bl s^{\delta})
=
\SC_A(\lime,\bl s^{\delta}) - \SC_A(\nim,\bl s^{\delta})
\]
where $A=\left(\frac{1}{2n},\frac{3}{2n}\right)\cup \left(\frac{2n-3}{2n},\frac{2n-1}{2n}\right)$. More elaborately, 
\lime{} coincides with \nim{} for every user outside of $A$, and hence it suffices to examine the social cost of the users in $A$ only. 

Under \lime{}, every user $t\in\left(\frac{1}{2n},\frac{3}{2n}\right)$ covers a distance of $\frac{2n-1}{2n}-t -\delta$. By integrating over all the users in $\left(\frac{1}{2n},\frac{3}{2n}\right)$ we get 
\begin{align}
&\int_{\nicefrac{1}{2n}}^{\nicefrac{3}{2n}}\left(\frac{2n-1}{2n}-\delta-t\right)dt=\left[\left(\frac{2n-1}{2n}-\delta\right)t -\frac{t^2}{2} \right]\bigg \mid_{\nicefrac{1}{2n}}^{\nicefrac{3}{2n}} \nonumber \\
& = \left(\frac{2n-1}{2n}-\delta\right)\frac{3}{2n} -\frac{9}{8n^2}-\left(\frac{2n-1}{2n}-\delta\right)\frac{1}{2n}+\frac{1}{8n^2}\nonumber \\
& = \frac{2n-3}{2n^2}-\frac{\delta}{n}.
\end{align}
In addition, w.p. $\epsilon$ every user in the interval is served by a random facility, since the PII  $\left(\frac{1}{2n},\frac{3}{2n}\right)$ does not have a facility from its left (Line \ref{alg:lime l or r} of Algorithm \ref{alg:lime}). $\bl s$ contains only two distinct locations; thus, w.p. $\frac{\epsilon\ceil{\frac{n}{2}}}{n}$ \lime{} directs every $t\in\left(\frac{1}{2n},\frac{3}{2n}\right)$ as \nim{}. Due to symmetry (of the uniform distribution), the users in $\left(\frac{2n-3}{2n},\frac{2n-1}{2n}\right)$ cover the same distance. 
Overall,
\begin{equation}
\label{eq:lower bound lime ic 1}
\SC_A(\lime,\bl  s^\delta) = \left( 1-  \frac{\epsilon}{2} \right)\left( \frac{2n-3}{n^2}-\frac{2\delta}{n} \right)+  \frac{\epsilon}{2}\SC_A(\nim,\bl  s^\delta).
\end{equation}

On the other hand, under $\nim$ and the profile $\bl  s^\delta$, the users in $\left(\frac{1}{2n},\frac{3}{2n}\right)$ travel a total distance of
\begin{align}
&\int_{\nicefrac{1}{2n}}^{\nicefrac{1}{2n}+\delta} (\frac{1}{2n}+\delta-t)dt + \int_{\nicefrac{1}{2n}+\delta}^{\nicefrac{3}{2n}}(t-(\frac{1}{2n}+\delta))dt \nonumber\\
& = \left[\left(\frac{1}{2n}+\delta  \right)t-\frac{t^2}{2} \right]\bigg \mid_{\nicefrac{1}{2n}}^{\nicefrac{1}{2n}+\delta}+\left[\frac{t^2}{2}-t(\frac{1}{2}+\delta) \right]\bigg \mid^{\nicefrac{3}{2n}}_{\nicefrac{1}{2n}+\delta}  \nonumber\\
& = \frac{1}{2}\left(\frac{1}{2n}+\delta \right)^2-\frac{1}{4n^2}-\frac{\delta}{2n}+\frac{1}{8n^2}+\frac{9}{8n^2}-\frac{3}{2n}\left(  \frac{1}{2n}+\delta \right)+\frac{1}{2}\left(\frac{1}{2n}+\delta \right)^2\nonumber\\
& = \left(\frac{1}{2n}+\delta  \right)^2-\frac{1}{8n^2}-\frac{\delta}{2n}+\frac{3}{8n^2}-\frac{3\delta}{2n} \nonumber\\
& =\left(\frac{1}{2n}+\delta  \right)^2+\frac{1}{4n^2}-\frac{2\delta}{n} \nonumber \\
& =\frac{1}{2n^2}-\frac{\delta}{n}+\delta^2\nonumber.
\end{align}
A symmetric argument applies for the users in $\left(\frac{2n-3}{2n},\frac{2n-1}{2n}\right)$ ; thus
\begin{equation}
\label{eq:lower bound lime ic 3}
\SC_A(\nim,\bl  s^\delta) = \frac{1}{n^2}-\frac{2\delta}{n}+2\delta^2.
\end{equation}
By combining Equations (\ref{eq:lower bound lime ic 1}) and (\ref{eq:lower bound lime ic 3}) we get
\begin{equation}
\label{eq:lower bound lime ic before final}
\SC_A(\lime,\bl  s^\delta)-\SC_A(\nim,\bl  s^\delta)= \left(1-  \frac{\epsilon}{2}\right)\left(\frac{2n-4}{n^2}-2\delta^2\right).
\end{equation}
The proof is completed by taking the limit,
\[
\IC_n(\lime)=\sup_{\bl s\in S}  \left\{  \SC(\lime,\bl s) - \SC(\nim,\bl s) \right\} \geq \lim_{\delta \rightarrow 0^+} \left(\SC(\lime,\bl  s^\delta)-\SC(\nim,\bl  s^\delta)\right)=\left(1-  \frac{\epsilon}{2}\right)\left(\frac{2n-4}{n^2}\right).
\]
\end{proof}
\section{Proof of Theorem  \ref{thm:LIME-IC}}
\begin{proof}[Proof of Theorem \ref{thm:LIME-IC}]

Let $o^n_i = \frac{2i-1}{2n}$ for $i\in [n]$, and denote
\[
I_{\bl s} \defeq \left\{j\in [n-1] : \bl s\cap \left(o^n_{j},o^n_{j+1} \right) \neq \emptyset   \right\}  , \quad 
A_{\bl s} \defeq \bigcup_{j\in I_{\bl s}} \left(o^n_{j},o^n_{j+1} \right).
\]
By definition\footnote{to ease readability, we omit the analysis of user \lime{} directs to random facilities, as written in the body of the paper.} of \lime{}, it suffices to examine all users $t$ such that $t\in A_{\bl s}  $; hence, it holds that 
\begin{align} 
\label{eq: lime ic main 1}
\IC_n(\lime)&=\sup_{\bl s \in S} \left\{ \SC(\lime,\bl s)-\SC(\nim,\bl s) \right\} \nonumber \\
&=\sup_{\bl s \in S}\left\{ \int_{t\in A_{\bl s}}  \sum_{i=1}^n\lime(\bl s,t)_i \cdot\abs{s_i-t}dt -\int_{t\in A_{\bl s}}\min_{i' \in [n]} \abs{s_i'-t}  dt  \right\}.
\end{align}
Let 
\[
f(\bl s)\defeq \int_{t\in A_{\bl s}}  \sum_{i=1}^n\lime(\bl s,t)_i \cdot\abs{s_i-t}dt -\int_{t\in A_{\bl s}}\min_{i' \in [n]} \abs{s_i'-t}  dt.
\]
We proceed by showing that $f(\bl s) \leq \frac{2n-3.5}{n^2}$ for every $\bl s \in S$. 

$\bullet$ First, we analyze the cases of $\abs{I_{\bl s}}=0$ and $\abs{I_{\bl s}}=1$. Line \ref{alg:lime outside} in Algorithm \ref{alg:lime} implies that for every $\bl s$ such that  $\abs{I_{\bl s}}=0$, \lime{} operates exactly as \nim{}; hence, $f(\bl s)=0$ for such profiles.
In case $\abs{I_{\bl s}}=1$, denote the only element in $I_{\bl s}$ by $j$. If $\bl s \cap \left(o^n_{j},o^n_{j+1} \right) = \bl s$, then due to Line \ref{alg:lime all in one} in Algorithm \ref{alg:lime}   \lime{} operates as \nim{}, and hence $f(\bl s)=0$.
Otherwise, $\bl s \cap \left(o^n_{j},o^n_{j+1} \right) \neq \emptyset$, and \lime{} directs all the users in $\left(o^n_{j},o^n_{j+1} \right)$ to their nearest facility outside that interval (Line \ref{alg:lime l and r}). Since we wish to bound $f$, we can w.l.o.g. assume that $\bl s \setminus \left(o^n_{j},o^n_{j+1} \right)$  contains a single facility, as otherwise $f$ can only be smaller; thus, denote the location of the only facility outside $\left(o^n_{j},o^n_{j+1} \right)$ by $x$. It holds that
\[
f(\bl s)=\int\displaylimits_{t\in (o^n_{j},o^n_{j+1})} \left(\abs{t-x}  -\min_{s_i \in\bl s} \abs{s_i-t}\right)dt \leq \int\displaylimits_{t\in (o^n_{j},o^n_{j+1})}\abs{t-x}dt -\frac{1}{4n^2}.
\]
Due to monotonicity, the worst case (i.e., highest value of $f$) is obtained for $x=1$ and $j=1$, or equivalently $x=0$ and $j=n-1$; thus,
\begin{align*}
f(\bl s) \leq \left(\frac{1}{n}\right)^2\frac{1}{2} +\frac{1}{n}\left( 1-\frac{3}{2n} \right)-\frac{1}{4n^2}  = \frac{2n-2.5}{2n^2} \leq \frac{2n-3.5}{n^2}.
\end{align*}
$\bullet$ We are left to analyze the case of $\abs{I_{\bl s}}\geq 2$, which turns to be trickier. Notice that 
\begin{equation}
\label{eq:using opt to bound ic ads}
\IC_n(\lime) <  \sup_{\bl s \in S} \left\{\int_{t\in A_{\bl s}}  \sum_{i=1}^n\lime(\bl s,t)_i \cdot\abs{s_i-t} dt -\abs{I_{\bl s}}\frac{1}{4n^2}\right\},
\end{equation}
where the last inequality follows from the bound on optimal social cost, $\nicefrac{1}{4n}$. Due to the definition of \lime{}, 
\begin{equation}
\label{eq: lime ic main 1.5}
\text{RHS of Eq. }(\ref{eq:using opt to bound ic ads})=
\sup_{\bl s \in S}\left\{\sum_{j\in I_{\bl s}}\mkern5mu\int\displaylimits_{t\in (o^n_{j},o^n_{j+1})}   \min_{i \in\bl s\setminus (o^n_{j},o^n_{j+1})} \abs{s_i-t}dt-\abs{I_{\bl s}}\frac{1}{4n^2}\right\} .
\end{equation}
Clearly, if under a strategy profile there are facilities outside the PIIs, the term inside the supremum can only decrease; thus, it is suffice to consider strategy profiles with facilities inside PIIs only. Next,

\begin{proposition}
\label{prop:median in interval}
Let $\bl s$ be a strategy profile, and let $(a,b)$ be an interval such that $\bl s \cap (a,b)=\emptyset$. It holds that 
\[
\int_a^b \min_{i\in [n]}\abs{s_i-t}dt \leq (b-a)\min_{i\in [n]}\abs{s_i-\frac{a+b}{2}}.
\]
\end{proposition}
The proof of this proposition appears below. By invoking Proposition \ref{prop:median in interval} we have 
\begin{align}
\label{eq: lime ic main 2}
\text{Eq. (\ref{eq: lime ic main 1.5})} &\leq \sup_{\bl s \in S}\sum_{j\in I_{\bl s}}\frac{1}{n} \min_{i \in I_{\bl s}\setminus \{j\}} \left(\abs{s_i-\frac{j}{n}}-\frac{1}{4n} \right) \nonumber \\
& \leq \sup_{\bl s \in S}\sum_{j\in I_{\bl s}} \frac{1}{n}\min_{i \in I_{\bl s}\setminus \{j\} } \left(\abs{\frac{i}{n}-\frac{j}{n}} +\frac{1}{2n}-\frac{1}{4n}  \right),
\end{align}
where the last inequality is due to triangle inequality, since the distance between every facility $s_i$ and the center of the interval $\left(\frac{2l-1}{2n},\frac{2l+1}{2n}  \right)$ it relies in is at most half of the interval size. Since there is an onto mapping from $S$ to $2^{[n-1]}$ (as there $n$ players but $n-1$ intervals),
\begin{align}
\label{eq: lime ic main 3}
\text{Eq. (\ref{eq: lime ic main 2})} = \max_{D \subseteq [n-1]}\sum_{j\in D} \frac{1}{n}\min_{i \in D\setminus \{j\} } \left(\abs{\frac{i}{n}-\frac{j}{n}} +\frac{1}{4n} \right)=\frac{1}{n^2}\max_{D \subseteq 2^{[n-1]}} \sum_{j\in D}\min_{i \in D\setminus \{j\} }\left( \abs{i-j}+\frac{1}{4} \right).
\end{align}
Next, we prove the following lemma:
\begin{lemma}
\label{lemma:aux for lime ic}
For every $D \subseteq {[n-1]}$, it holds that 
\[
\sum_{j\in D}\min_{i \in D\setminus \{j\} }\left( \abs{i-j}+\frac{1}{4} \right)\leq 2n-3.5
\]
\end{lemma}
The proof of this lemma appears below. By combining Lemma \ref{lemma:aux for lime ic} with Equation (\ref{eq: lime ic main 3}), we obtain
\begin{align}
\label{eq: lime ic main 4}
\text{Eq. (\ref{eq: lime ic main 3})} \leq  \frac{2n-3.5}{n^2}.
\end{align}
Overall, we conclude that $ \IC_n(\lime) < \frac{2n-3.5}{n^2}$.
\end{proof}

\begin{proof}[Proof of Proposition \ref{prop:median in interval}]
Let $\alpha$ such that $\alpha =\argmin_{s_i \in \bl s}\abs{s_i-\frac{a+b}{2}}$ for some $\alpha \geq 0$. Since $\bl s \cap (a,b)=\emptyset$,  $\alpha \notin (a,b)$  holds, so assume w.l.o.g. that $\alpha \geq b$. We continue by highlighting two cases. In case all the users in $(a,b)$ are served by the facility in $\alpha$, i.e., $\argmin_{s_i \in \bl s}\abs{s_i-t}=\alpha$ for every $t\in(a,b)$, we have
\begin{align*}
\int_a^b \min_{i\in [n]}\abs{s_i-t}dt &=\int_a^b \left(  \alpha-t\right)dt = (b-a)\alpha-\left(\frac{b^2-a^2}{2}\right)=(b-a)\left(  \alpha-\frac{a+b}{2}\right).
\end{align*}
Otherwise, if for some values of $t\in (a,b)$ it holds that $\beta =\argmin_{s_i \in \bl s}\abs{s_i-t}$ for $\beta\neq \alpha$ (and therefore $\beta \leq a$), we have
\begin{align*}
\int_a^b \min_{s_i \in \bl s}\abs{s_i-t}dt &=\int_a^b \min \{\abs{\alpha-t},\abs{\beta - t}\}dt<\int_a^b \left(  \alpha-t\right)dt ,
\end{align*} 
which is identical to the previous case.
\end{proof}

\begin{proof}[Proof of Lemma \ref{lemma:aux for lime ic}]
We solve the optimization problem by a reduction to graph problem. Consider the directed graph $G=(V,E)$, where $V=\{v_1,\dots,v_{n-1}\}$, and define 
\[
E=\bigcup_{i\in [n-2]}\left\{ (v_i,v_{i+1}),(v_{i+1},v_{i}) \right\}.
\]
Next, observe that the function given in the lemma is equivalent to the selection of $D\subseteq V$ as to maximize
\[
f(D)\defeq\sum_{v_j\in D}\min_{v_i \in D\setminus \{v_j\} }\left( \abs{i-j}+\frac{1}{4} \right)=\frac{\abs{D}}{4}+\sum_{v_j\in D}\min_{v_i \in D\setminus \{v_j\} } \abs{S(v_i,v_j)},
\]
where $S(v_i,v_j)$ is the set of edges corresponding to the shortest path from $v_i$ to $v_j$. Let $D^* = \{v_1,v_{n-1}\}$, and observe that $f(D^*)= \abs{E}+\frac{1}{2}$. In addition, notice that by selecting a vertex $v_i$ for $i\in \{2,3,\dots n-2\}$,  $S(v_i,v_j)$ contains exactly one of $\{(v_i,v_{i+1}),(v_i,v_{i-1})\}$. As a result, for every internal (i.e. one of $\{v_2,\dots,v_{n-2}  \}$) $D$ contains, it loses one edge from $E$; thus, for every $D$ of size greater than two (that must contain at least one of $\{v_2,\dots,v_{n-2}  \}$),
\[
f(D) \leq \frac{\abs{D}}{4}+ \abs{E} -(\abs{D}-2)=\abs{E}-\frac{3\abs{D}}{4}+2 \stackrel{\abs{D}>2}{<}\abs{E}<f(D^*).
\]
This completes the proof of the lemma.
\end{proof}

\section{Omitted proofs from Section \ref{sec:extensions}}

\begin{proof}[Proof of Proposition \ref{prop:two-player eq}]
This proposition is a special case of Proposition \ref{prop:2 players lambda eq}.
\end{proof}

\begin{claim}
\label{claim:no neutral}
There is no neutral mediator such that
\begin{itemize}
\item $(x,y)$ is a PNE for some $x,y \in [0,1]$ such that $x\neq y$
\item For every $z$, $(z,z)$ is not a PNE
\end{itemize}
\end{claim}

\begin{proof}[Proof of Claim \ref{claim:no neutral}]
Assume by contradiction that $(x,y)$ is a PNE; hence,
\[
\pi_1(x,y) \geq \pi_1(y,y)\stackrel{\text{Neutrality}}{=} \frac{1}{2}.
\]
Due to symmetry and the game being constant sum, $\pi_1(x,y)=\pi_2(x,y)=\frac{1}{2}$. In addition, neutrality suggests that $\pi_1(y,y)=\frac{1}{2}$ and due to the second condition $(y,y)$ in not a PNE; thus, there exists $x'$ such that
\[
\pi_1(x',y) > \pi_1(y,y) = \frac{1}{2} =\pi_1(x,y),
\]
hinting that $(x,y)$ is not a PNE, and we obtained a contradiction.
\end{proof}

\begin{proof}[Proof of Lemma \ref{lemma:glime pne}]
Let $\bl s^* \defeq (q_{\nicefrac{2i-1}{2n}})_{i=1}^{n}$. The proof consists of two steps: first, we show that $\bl s^*$ is a PNE; then, we show it is the unique PNE of the induce game.

\noindent \textit{Step 1:}  Notice that the facility in $q_{\nicefrac{1}{2n}}$ obtains a payoff of exactly $\frac{1}{n}$, with a user mass of the interval $\left(q_{\nicefrac{1}{2n}},q_{\nicefrac{3}{2n}} \right)$ arriving from the right (see Line \ref{alg:glime:half half} in Algorithm \ref{alg:glime}), which is $\frac{1}{2n}$, and a user mass of $\frac{1}{2n}$ arriving from the left. The same applies for the facility located in $q_{\nicefrac{2n-1}{2n}}$. In addition, for every $i\in \{2,\dots,n-1\}$, the facility located in $q_{\nicefrac{2i-1}{2n}}$ gets half of the user mass in $\left( q_{\nicefrac{2i-3}{2n}},q_{\nicefrac{2i-1}{2n}} \right)$, and half of the user mass of in $\left( q_{\nicefrac{2i-1}{2n}},q_{\nicefrac{2i+1}{2n}} \right)$. Overall, each player gets exactly $\frac{1}{n}$.

Consider a deviation of player $j$ for some arbitrary $j$. If player $j$ would deviate to a point in $\left[0,q_{\nicefrac{1}{2n}} \right) $ or $\left(q_{\nicefrac{2n-1}{n}},1\right]$, she could get at most $\frac{1}{2n}+\frac{\epsilon}{n^2}$ (the term $\frac{\epsilon}{n^2}$ arriving from the case $j=1$ or $j=n$). If she would deviate to a point in $\left( q_{\nicefrac{2i-1}{2n}},q_{\nicefrac{2i+1}{2n}} \right)$ for some $i\in[n-1]$, she would get at most $\frac{1}{2n}+\frac{\epsilon}{n^2}$. Finally, if deviates to a point $q_{\nicefrac{2i-1}{2n}}$ for some $i\in[n]$, her payoff could be at most $\frac{3}{4n}+\frac{\epsilon}{n^2} $. Consequently, she has no beneficial deviation. Overall, we showed that $\bl s^*$ is a PNE.

\noindent \textit{Step 2:} Let $\bl s=(s_1,\dots,s_n)$ be an arbitrary strategy profile such that w.l.o.g. $s_1\leq \cdots \leq s_n$, $\bl s\neq \bl s^*$, and assume by contradiction that $\bl s$ is in equilibrium. First, we show that

\begin{lemma}
\label{glime: notininterval}
Consider $G(n,\glime{})$ for any $n\geq 2$. If $\bl s$ is an equilibrium profile, then $s_j \notin\left(q_{\frac{2i-1}{2n}},q_{\frac{2i+1}{2n} } \right)$ for every $j\in [n]$ and $i\in [n-1]$.
\end{lemma}
The proof of the Lemma \ref{glime: notininterval} follows closely the proof of Lemma \ref{notininterval}, by replacing $o^n_i$ with its counterpart for non-uniform distributions, $q_{\frac{2i-1}{2n}}$; hence, it is omitted. This equivalence stems from the proof of Lemma \ref{notininterval} considering beneficial deviations to $\cup_{i=1}^n\{o^n_i\}$, which yield the same payoffs as deviations to $\cup_{i=1}^n\{q_{\nicefrac{2i-1}{2n}}\}$ under a non-uniform distribution.

Next, by reiterating the steps given in the proof of Theorem \ref{thm: lime unique pne} and replacing $o^n_i$ with $q_{\frac{2i-1}{2n}}$, we show that under $\bl s \cap \bl s^* = \bl s$. Namely, we show that if $\bl s$ is a PNE then the facilities' locations under $\bl s$ are a subset of $\cup_{i=1}^n\{q_{\nicefrac{2i-1}{2n}}\}$, the locations selected in $\bl s^*$. The proof is completed by showing that each location in $\cup_{i=1}^n\{q_{\nicefrac{2i-1}{2n}}\}$ is selected by exactly one player, which is also the final step in the proof of Theorem \ref{thm: lime unique pne}.
\end{proof}

\begin{proof}[Proof of Proposition \ref{lemma:glime ic lower}]
Consider the profile $\bl s$ where half of the players select $\frac{1}{2n}$ and the other half select $\frac{2n-1}{n}$. It holds that 
\begin{align*}
&\SC(\glime{},\bl s)-\SC(\nim{},\bl s) \\
&=\frac{1}{2}\left(\frac{n-1}{n}  \right)^2-\left(\frac{1}{2}-\frac{1}{2n}  \right)^2 \\
&=\frac{1}{2}-\frac{1}{n}+\frac{1}{2n^2}-\left( \frac{1}{4}-\frac{1}{2n}+\frac{1}{4n^2} \right)\\
&=\frac{1}{4}-\frac{1}{2n}+\frac{1}{4n^2}.
\end{align*}
\end{proof}

\begin{corollary}
Consider $G(n,\lime{},g)$ for $n\geq 3$. The social cost of the unique PNE is 
\[
{ \footnotesize
\frac{1}{2}\left(q^2_{\nicefrac{1}{2n}} +q^2_{\nicefrac{2n-1}{2n}}  +\sum_{i=1}^{n-1} (q_{\nicefrac{2i+1}{2n}}-q_{\nicefrac{2i-1}{2n}})^2 \right).
}
\] 
\end{corollary}

For the uniform user distribution, the social cost is $\nicefrac{2n-1}{4n^2}$, which is relatively high when compared to social cost of \lime{}. Nevertheless, recall that \glime{} is applicable for \textit{any} general distribution.

\begin{proof}[Proof of Lemma \ref{lemma:glime ic upper}]

Let $o^n_i = \frac{2i-1}{2n}$ for $i\in [n]$, and denote $I\defeq\cup_{i\in[n-1]}\left(o^n_i,o^n_{i+1}\right)$. Recall that
\[
\SC(\glime,\bl s)-\SC(\nim,\bl s)= \int_0^1\left( \sum_{i=1}^n\glime(\bl s,t)_i \cdot\abs{s_i-t} -\min_{i' \in [n]} \abs{s_{i'}-t} \right)dt.
\]
For every $t\notin I$  it holds that
\[
\sum_{i=1}^n\glime(\bl s,t)_i \cdot\abs{s_i-t} =\min_{i' \in [n]} \abs{s_{i'}-t},
\]
since $\glime$ directs it to its nearest facility under $\bl s$; hence, it suffices to consider users in for which \glime{} operates non-optimally, i.e. analyze $f$ that is define as follows:
\[
f(\bl s)\defeq \int_{I}\left( \sum_{i=1}^n\glime(\bl s,t)_i \cdot\abs{s_i-t} -\min_{i' \in [n]} \abs{s_{i'}-t} \right)dt.
\]
In addition, let $L(\bl s,t)$ denote the distance $t$ covers when \glime{} sends him to its nearest facility from the left, and similarly let $R(\bl s,t)$ denote the distance when he is send by \glime{} to the right; thus,
\begin{align}
\label{eq: ic for glime after approx}
f(\bl s)= \int_{I}\left( \frac{1}{2} L(\bl s,t)+\frac{1}{2}R(\bl s ,t) -\min_{i' \in [n]} \abs{s_{i'}-t} \right)dt.
\end{align}
For every $i\in [n-1]$ and a user $t\in(o^n_i,o^n_{i+1})$ it holds that
\[
\frac{1}{2} L(\bl s,t)+\frac{1}{2}R(\bl s ,t) =  \frac{1}{2} L\left(\bl s,\frac{i}{n}\right)+\frac{1}{2}R\left(\bl s ,\frac{i}{n}\right);
\] 
hence,
\begin{align}
\label{eq:lemma glime two term}
f(\bl s) &= \underbrace{\sum_{i=1}^{n-1}\frac{1}{2n} \left(L\left(\bl s ,\frac{i}{n}\right)+R\left(\bl s ,\frac{i}{n}\right) \right)}_{\text{term I}}- \underbrace{\int_{I}\left( \min_{i' \in [n]} \abs{s_{i'}-t} \right)dt}_{\text{term II}}.
\end{align}
Next, we analyze the two terms in Equation (\ref{eq:lemma glime two term}):
\begin{itemize}
\item Notice that 
\[
\text{term I} = \sum_{i=1}^{n-1}\frac{1}{2n} \left(L\left(\bl s ,\frac{i}{n}\right)+R\left(\bl s ,\frac{i}{n}\right) \right) \leq \frac{n-1}{2n},
\]
where equality holds only if $\bl s \cap \{0,1\}=\bl s$.
\item Due to Claim \ref{koneplayer_opt}, we know that 
\[
\text{term II} = \int_{I}\left( \min_{i' \in [n]} \abs{s_{i'}-t} \right)dt  \geq \abs{I} \frac{1}{4n},
\]
where equality holds only under the (unique, up to renaming the players) strategy profile $\bl s_0$ under which the facilities are located in the optimal locations w.r.t. to $[o^n_1,o^n_{n}]$. In particular, $\bl s_0 \cap \{0,1\} \neq \bl s_0$, so under $\bl s_0$ term I is strictly less than its bound.
\end{itemize}
Combining what we know on terms I and II, we conclude that for every profile $\bl s$ 
\begin{align}
\label{eq:lemma glime}
f(\bl s)< \frac{n-1}{2n}- \frac{n-1}{n}\cdot \frac{1}{4n}  
=\frac{1}{2}-\frac{3}{4n}+\frac{1}{4n^2},
\end{align} 
which is less than the lower bound obtained for \dict{} (see Lemma \ref{lemma:dict high intervention cost}).
\end{proof}

\section{Configurable limited intervention mediator}
\label{sec:configurable}
\begin{algorithm}[t]
\caption{Configurable Limited Intervention Mediator \label{alg:clime}}
\DontPrintSemicolon
 \KwIn{ A strategy profile $\bl s$, user $t$ and a hyper-parameter $\lambda$}
 \KwOut{A location in $\bl s$}
  \uIf{ $\exists i\in \{1,n-1\}$ such that $t \in \left(\frac{i}{n}-\lambda,\frac{i}{n}+\lambda  \right)$}{
    $\bl l \gets \bl s \cap \left[0,\frac{i}{n}-\lambda\right], \bl r \gets \bl s \cap \left[\frac{i}{n}+\lambda,1  \right]$ \;
		\uIf{$\bl l \neq \emptyset$ and $\bl r \neq \emptyset$}
		{\KwRet $\nim (\bl l \cup \bl r, t)$}
		\uElseIf{
		$\bl l \neq \emptyset$ or $\bl r \neq \emptyset$
		}    
		{
		 w.p. $1-\epsilon$ \KwRet $\nim(\bl l \cup \bl r, t)$, otherwise \KwRet $random(\bl s,t)$
		}
		\uElse{\tcp{all facilities are inside that interval}
		\KwRet $\nim(\bl s, t)$ \label{alg:clime all in one}
		}
  }
  \uElse{\tcp{$t$ is outside the intervals}
    \KwRet $\nim(\bl s, t)$ \label{alg:clime as nime}\; 
  }  
\end{algorithm}
In this section we propose a mediator that uncovers low social cost with low intervention cost. The mediator, which we term Configurable Limited Intervention Mediator or $\clime(\lambda)$ for abbreviation, is given in Algorithm \ref{alg:clime}. $\clime(\lambda)$ operates similarly to \lime{}, but does two things differently. First, it contains one PII if $n=2$ (a two-player game), and two PIIs otherwise. Second, it receives as input a hyper-parameter, $\lambda$, which limits the length of the PIIs. This allows $\clime(\lambda)$ to tune the trade-off between low social cost and low intervention cost. When we discuss properties of the $\clime(\lambda)$ mediator that hold for general $\lambda$, we omit $\lambda$ and refer to the mediator as \clime{}. We begin by analyzing the two-player game, and then move to general games. 

\subsection{The two-player game}
In this subsection we formally analyze $G(2,\clime(\lambda))$. Since \clime{} is neutral (see Subsection \ref{subsec:neutral mediators}), Claim \ref{claim:no neutral} suggests that it either induces a game with a symmetric PNE, or a game with no PNE at all. We now show that the former is correct.
\begin{proposition}
\label{prop:2 players lambda eq}
Consider $G(2,\clime(\lambda))$ for $0<\lambda\leq\frac{1}{4}$. A strategy profile $\bl s=(s_1,s_2)$ is a PNE if and only if $s_1,s_2 \in \{\frac{1}{2}-\lambda, \frac{1}{2}+\lambda\}$.
\end{proposition}
In particular, by setting $\lambda = \frac{1}{4}$ we recover the results obtained for $G(2,\lime)$ in Proposition \ref{prop:two-player eq}. Next, denote the set of PNE by $E \defeq \{\frac{1}{2}-\lambda, \frac{1}{2}+\lambda\}^2$. By a simple calculation, one can see that
\begin{corollary}
\label{cor:sc for n=2}
Consider $G=G(2,\clime(\lambda))$ for $0<\lambda\leq\frac{1}{4}$.
for every $\bl e \in E$, it holds that $\SC(\clime(\lambda),\bl e) \in \{\frac{1}{4}-\lambda+2\lambda^2 ,\frac{1}{4}+\lambda^2  \}$.
\end{corollary}
Consequently, one can view that the best and worst social cost as a function of $\lambda$,
\[
\min_{\bl e \in E} \SC(\clime(\lambda),\bl e) = \frac{1}{4}-\lambda+2\lambda^2, \quad \max_{\bl e \in E} \SC(\clime(\lambda),\bl e) = \frac{1}{4}+\lambda^2
\]

Notably, the best social cost is $\frac{1}{8}$, obtained for $\lambda = \frac{1}{4}$. However, the social cost can also deteriorate (when compared to \nim{}), and is $\frac{5}{16}$ in the worst case. This does not come as a surprise, since for $\lambda = \frac{1}{4}$ the games $G(2,\clime(\frac{1}{4}))$ and $G(2,\lime)$ are strategically equivalent. 

We proceed to analyzing the intervention cost of \clime{}.
\begin{proposition}
\label{prop:clime n=2 ic}
It holds that $\IC_n (\clime({\lambda})) \leq \lambda - \lambda^2 $.
\end{proposition}
Now, the trade-off between the minimal social cost in equilibrium and the intervention cost can be quantified as follows,
\[
\min_{0<\lambda \leq \frac{1}{4}} \{ \min_{\bl e \in E} \SC(\clime(\lambda),\bl e) + c\cdot \IC_n (\clime({\lambda}))  \},
\]
where $c$ is a penalty for intervening. Equivalently,
\[
\min_{0<\lambda \leq \frac{1}{4}} \left\{ \frac{1}{4}+(c-1)\lambda + (2-c)\lambda^2  \right\};
\]
thus, $\clime({\lambda})$ is applicable whenever the mediator designer gives more weight to improving the social cost (low values of $c$, i.e. $c<1$). On the other hand, if $c\geq 1$, intervention is not desired, and the designer would be better off with adopting \nim{} as a mediator.
\subsection{Multi-player games}
We now focus on $G(n, \clime({\lambda}))$ for $n\geq 3$. First, we show a negative result for \clime{} for $n=3$.
\begin{proposition}
\label{prop:clime n = 3}
Consider $G(3,\clime(\lambda))$. For every $0<\lambda<\frac{1}{6}$, there is no PNE.
\end{proposition}

As it turns out, the equilibrium structure varies for $n\geq 4$. As mentioned above, we do not study it in this paper but rather exemplify some cases. Before we do so, we give a crude bound on the intervention cost of games with $n\geq 3$.
\begin{proposition}
\label{prop: ic of clime}
For $n\geq 3$, it holds that $\IC(\clime{}) \leq 4 \lambda$.
\end{proposition}
We now showcase a few examples.
\subsubsection{Four players}

Consider $G(4,\clime(\frac{1}{8}))$. The strategy profile $\bl s=(\frac{1}{8},\frac{3}{8},\frac{5}{8},\frac{7}{8})$ is the unique PNE (up to renaming the players) of this game, with a social cost of  $\SC(\clime,\bl s)=\frac{1}{16}$, i.e., the optimal social cost in a four-player game. In fact, this equilibrium structure extends to $G(4,\clime(\lambda))$ with any $\lambda \leq \frac{1}{8}$. See Figure \ref{fig:four players} for illustration.

\subsubsection{Five players}
For $n = 5$, and unlike $n=4$, $\clime({\lambda})$ may induce a game with several equilibria, depending on the value of $\lambda$. For example, if $\lambda = \frac{1}{10}$  the socially optimal locations (i.e.,$\frac{1}{10},\frac{3}{10},\frac{5}{10},\frac{7}{10},\frac{9}{10}$) are the unique equilibrium. However, when $\lambda = \frac{1}{40}$, agglomerating on the two distant endpoints of the PIIs can be in equilibrium. See Figure \ref{fig:four players} for illustration.

\subsubsection{Games with $n\geq 6$}
See Figure \ref{fig: greater n} for illustration. Here too, several equilibria may exist. To conclude, although not always the case, \clime{} illustrate in a constructive manner an approach for trading some social cost for getting lower intervention cost.

\begin{figure*}[h]
\centering
\includegraphics[scale=1.0]{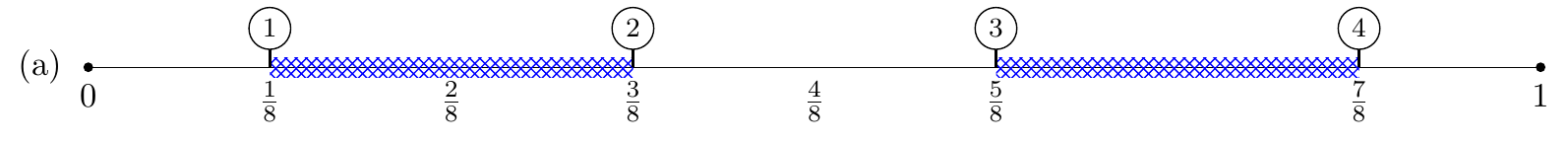}
\includegraphics[scale=1.0]{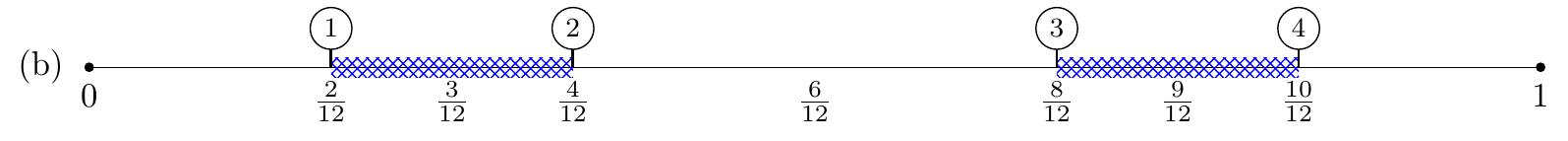}
\caption{Two PNE profiles in four-player games with \clime{}. Sub-figure (a) corresponds to G(4,$\clime({\lambda})$) with $\lambda = \frac{1}{8}$, and Sub-figure (b) corresponds to G(4,$\clime({\lambda})$) with $\lambda = \frac{1}{12}$. Each profile is the unique PNE of the game induces by $\clime({\lambda})$ with the corresponding $\lambda$.
\label{fig:four players}}
\end{figure*}
\begin{figure*}[h]
\centering
\includegraphics[scale=1.0]{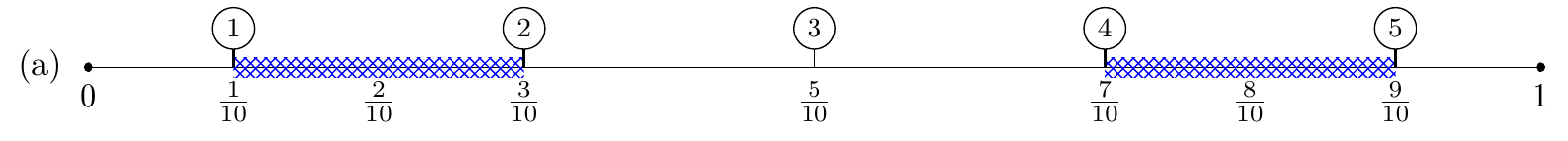}
\includegraphics[scale=1.0]{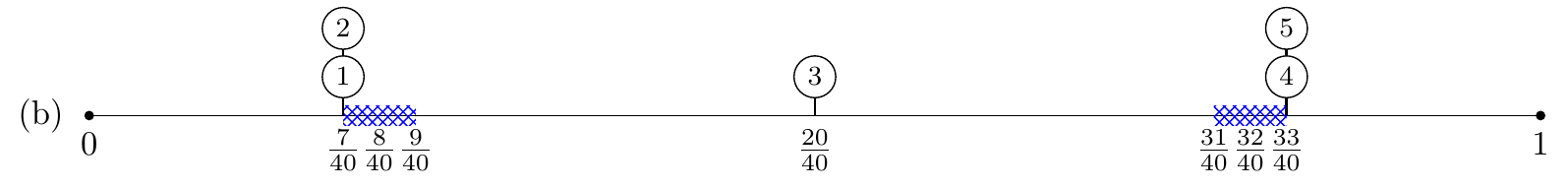}
\caption{Two PNE profiles in five-player games with \clime{}. Sub-figure (a) corresponds to G(4,$\clime({\lambda})$) with $\lambda = \frac{1}{10}$, and Sub-figure (b) corresponds to G(4,$\clime({\lambda})$) with $\lambda = \frac{1}{40}$. Note that in Sub-figure \ref{fig:five players}.(b) there are paired facilities, unlike the examples given for \lime{}.  \label{fig:five players}}
\end{figure*}
\begin{figure*}[h]
\centering
\includegraphics[scale=1.0]{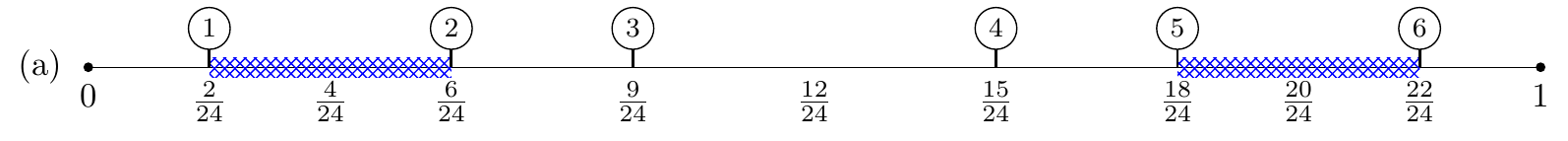}
\includegraphics[scale=1.0]{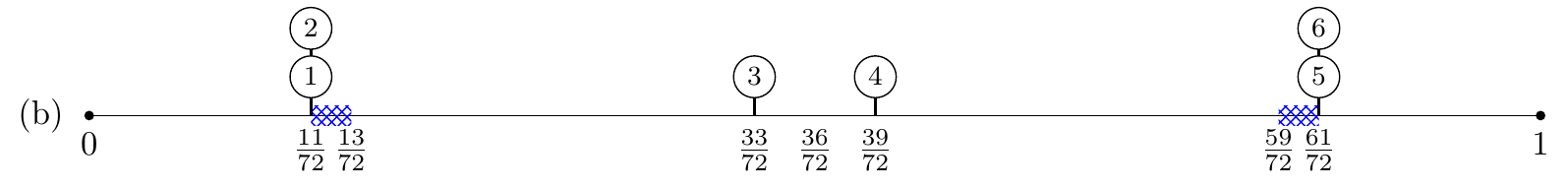}
\includegraphics[scale=1.0]{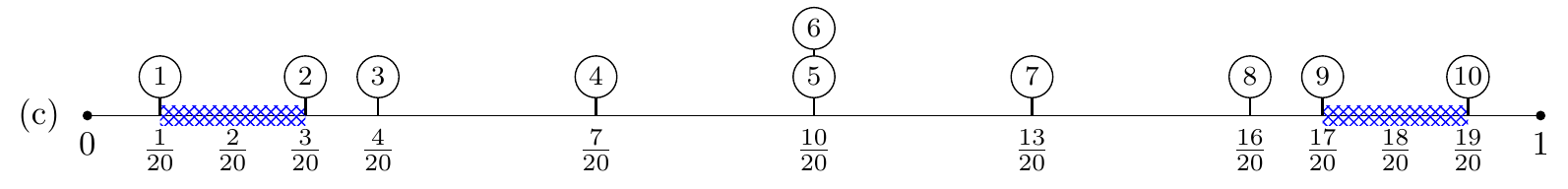}
\caption{Examples of PNE profile of $G(n,\clime({\lambda} ))$, for various values of $n$ and $\lambda$. In Sub-figure \ref{fig: greater n}.(a) $n=6$ and $\lambda = \frac{1}{12}$, in Sub-figure \ref{fig: greater n}.(b) $n=6$ and $\lambda = \frac{1}{72}$, and in Sub-figure \ref{fig: greater n}.(c) and $n=10$ and $\lambda = \frac{1}{20}$.\omer{we removed \ref{fig: greater n}.(d) depict two PNE for} 
\label{fig: greater n}}
\end{figure*}

\section{Proofs of Section \ref{sec:configurable}}
The following claim will be useful in what follows.
\begin{claim}
\label{claim:clime n=2}
Consider $G(2,\clime(\lambda))$ for $0<\lambda\leq\frac{1}{4}$. If a strategy profile $\bl s$ is a PNE, then $\bl s \cap \left( \frac{1}{2}-\lambda,\frac{1}{2}+\lambda \right) = \emptyset$.
\end{claim}

\begin{proof}[Proof of Claim \ref{claim:clime n=2}]
Assume the claim does not hold, and w.l.o.g. let $s_1\leq s_2$. If $\abs{\bl s \cap \left( \frac{1}{2}-\lambda,\frac{1}{2}+\lambda \right)}=2$, then by deviating to $s_2' = \frac{1}{2}+\lambda$ her payoff would increase, since
\[
\pi_2(s_1,s_2') = \left(1-\frac{\epsilon}{2}\right) 2\lambda +\frac{\epsilon\lambda}{2} + \frac{1}{2}-\lambda =\frac{1}{2}+\lambda\left(1-\frac{\epsilon}{2}\right)\stackrel{\epsilon <1}{>} \frac{1}{2} = \pi_2(\bl s);
\]
hence, $\bl s$ cannot be in equilibrium. Otherwise, w.l.o.g. $\bl s \cap \left( \frac{1}{2}-\lambda,\frac{1}{2}+\lambda \right)=s_2$. There are two cases:
\begin{enumerate}
\item If $s_1=\frac{1}{2}-\lambda$, then $\pi_2(\bl s)<\frac{1}{2}$ (as previously explained); thus, player 2 has a beneficial deviation
\item $s_1<\frac{1}{2}-\lambda$, then by deviating to $s_1'=\frac{1}{2}-\lambda$ player 1 would improve her payoff, since $\pi_1(\bl s_) < \pi_1(s_1',s_2)$.
\end{enumerate}
Under both cases the players have beneficial deviations; therefore, $\bl s$ is not in equilibrium.
\end{proof}

\begin{proof}[Proof of Proposition \ref{prop:2 players lambda eq}]
First, by a simple case analysis one can show that every profile in $\{ \frac{1}{2}-\lambda,\frac{1}{2}+\lambda \}^2$ is a PNE, since the players do not have beneficial deviations. On the other hand, let $\bl s = (x,y)$ such that $(x,y)\notin \{ \frac{1}{2}-\lambda,\frac{1}{2}+\lambda \}$, and assume by contradiction that $(x,y)$ is a PNE. In addition, w.l.o.g. assume that $x\notin \{ \frac{1}{2}-\lambda,\frac{1}{2}+\lambda \}$. Due to Claim \ref{claim:clime n=2}, we know that $x$ and $y$ are outside $\left( \frac{1}{2}-\lambda,\frac{1}{2}+\lambda  \right)$. If $x\neq y$, player 1 can move her facility towards that of player 2, and improve her payoff. Otherwise, if $x=y$, by deviating to one of $ \{ x-\delta,x+\delta \}$ for a small $\delta >0$, player 1 can obtain a payoff arbitrarily close to $\max \{ x ,1-x  \}$, which is greater than $\frac{1}{2}$ since $x \notin  \{ \frac{1}{2}-\lambda,\frac{1}{2}+\lambda \}$.
\end{proof}

\begin{proof}[Proof of Corollary \ref{cor:sc for n=2}]
For $\bl s_{\text{worst}} = \left( \frac{1}{2}-\lambda,\frac{1}{2}-\lambda  \right)$ (or equivalently, for $\bl s = \left( \frac{1}{2}+\lambda,\frac{1}{2}+\lambda  \right)$ ),
\[
\SC(\clime,s_{\text{worst}} ) = \left(\frac{1}{2}-\lambda \right)^2 \frac{1}{2}+\left(\frac{1}{2}+\lambda \right)^2 \frac{1}{2} = \frac{1}{4}+\lambda^2
\]
For $\bl s_{\text{best}} = \left( \frac{1}{2}-\lambda,\frac{1}{2}+\lambda  \right)$ it hold that
\[
\SC(\clime,s_{\text{best}} ) =  \left(\frac{1}{2}-\lambda \right)^2  +\lambda^2 =\frac{1}{4}-\lambda+2\lambda^2 
\]
\end{proof}

\begin{proof}[Proof of Proposition \ref{prop:clime n=2 ic}]
Let $\bl s=(x,y)$, and let
\[f\left(x,y\right) \defeq \SC\left(\clime,\left(x,y\right)\right) - \SC\left(\nim,\left(x,y\right)\right).\]
We are interested in 
\begin{equation}
\label{eq: prop ic 2 init}
\IC_2(\clime) = \sup_{x,y \in [0,1]} f(x,y)
\end{equation}
Let $I\defeq\left(\frac{1}{2}-\lambda, \frac{1}{2}+\lambda\right)$, and observe that if $x,y \in I$, then due to Line \ref{alg:clime all in one} in Algorithm \ref{alg:clime} it holds that $\clime(x,y,t)=\nim(x,y,t)$ for every user $t$. Moreover, if $x,y \notin I$, $\clime(x,y,t)=\nim(x,y,t)$ holds for every $t$ as well due to Line \ref{alg:clime as nime}. Consequently, 
\[
\text{Eq. } (\ref{eq: prop ic 2 init})=\sup_{x\notin I,  y\in I} f(x,y)=\sup_{x\leq \frac{1}{2}-\lambda,  y\in I} f(x,y).
\]
where the last equality is due to symmetry. Let $x,y$ such that $x\leq \frac{1}{2}-\lambda$ and $y\in I$, and notice that
\begin{equation}
\label{eq: prop ic 2 full}
f(x,y) = \int_{I} \left( \abs{t-x}-\min\{ \abs{t-y}, \abs{t-x}  \}  \right)dt.
\end{equation}
Next, we define $I_y(x,y)$ to be the set of users in $I$ that are closer to $y$ than to $x$, i.e., $I_y(x,y) = \{t\in I\mid \abs{t-y}\leq \abs{t-x}   \}$, and let $I_x(x,y) =I \setminus I_y(x,y)$. We can rephrase Equation (\ref{eq: prop ic 2 full}) as follows
\begin{align}
\label{eq: prop ic 2 short}
f(x,y) &= \int_{I_x(x,y)} \left( \abs{t-x}-\min\{ \abs{t-y}, \abs{t-x}  \}  \right)dt+\int_{I_y(x,y)} \left( \abs{t-x}-\min\{ \abs{t-y}, \abs{t-x}  \}  \right)dt\nonumber \\
&=\int_{I_y(x,y)} \left( \abs{t-x}-\abs{t-y} \right)dt.
\end{align}
We claim that $f$ is monotonically decreasing in $x$ when $y$ is fixed. This holds not only because the integrand in the last term in Equation (\ref{eq: prop ic 2 short}) is monotonically decreasing in $x$ when $y$ is fixed, but also since the set $I_y(x,y)$ monotonically decreases in $x$. As a result, 
\begin{align}
\label{sup with x=0 and y}
\text{Eq. } (\ref{eq: prop ic 2 init})=\sup_{y\in I} f(0,y).
\end{align}
Since $\lambda \leq \frac{1}{4}$, under every profile $(0,y)$ such that $y\in I$, \nim{} directs every $t\in I$ to $y$. In contrast, under every profile $(0,y)$ such that $y\in I$, \clime{} directs every $t\in I$ to $x$; thus, by a direct calculation we conclude that for every $y\in I$
\begin{align}
f(0,y) = \frac{(2\lambda)^2}{2}+2\lambda\left(\frac{1}{2}-\lambda\right)-\frac{1}{2} \left( \left(y-\frac{1}{2}+\lambda  \right)^2 +\left(\frac{1}{2}+\lambda-y  \right)^2  \right).
\end{align}
Moreover,
\begin{align}
\label{eq: ic 2:derivative}
\frac{d f(0,y)}{dy} = -\frac{1}{2} \left( 2\left(y-\frac{1}{2}+\lambda  \right) -2\left(\frac{1}{2}+\lambda-y  \right)  \right)=1-2y, \quad \frac{d^2 f(0,y)}{dy^2}=-2<0;
\end{align}
hence, $f(0,y)$ is concave in $y$, and attains its maximum in $y=\frac{1}{2}$. Since
\[
f\left(0,\frac{1}{2}\right) = \frac{(2\lambda)^2}{2}+2\lambda\left(\frac{1}{2}-\lambda\right)-\frac{1}{2} \left( \lambda^2+\lambda^2 \right)= \lambda- \lambda^2,
\]
due to Equation (\ref{sup with x=0 and y}) it holds that
\[
\IC_2(\clime) = \lambda- \lambda^2
\]
This concludes the proof of the Proposition. 
\end{proof}

\begin{claim}\label{claim:clime n=3}
Consider $G(3,\clime(\lambda))$ for $0<\lambda < \frac{1}{6}$. If $\bl s \cap \left(\left( \frac{1}{3}-\lambda,\frac{1}{3}+\lambda \right) \cup \left(\frac{2}{3}-\lambda,\frac{2}{3}+\lambda \right)\right) \neq \emptyset$ then $\bl s$ is not a PNE. 
\end{claim}
\begin{proof}[Proof of Claim \ref{claim:clime n=3}] Assume the claim does not hold, and w.l.o.g. $s_1\leq s_2 \leq s_3$. We show by exhaustive case analysis that $\bl s$ is not a PNE.
\begin{enumerate}
\item If $|\bl s \cap \left(\frac{1}{3}-\lambda,\frac{1}{3}+\lambda \right)|=3$: Player 1 obtains a payoff of less than $\frac{1}{3}+\lambda$, but by deviating to $\frac{1}{3}+\lambda$ she could get at least $\left(\frac{2}{3}+\lambda\right)(1-\epsilon)$; hence, she has a beneficial deviation (deviating to the left).
\item  If $|\bl s \cap \left(\frac{1}{3}-\lambda,\frac{1}{3}+\lambda \right)|=2$: In case $s_1,s_2 \in (\frac{1}{3}-\lambda,\frac{1}{3}+\lambda)$ and $s_3 > \frac{1}{3}+\lambda$, player 3 has a beneficial deviation. If $s_3=\frac{1}{3}+\lambda$ then player 1 moves to $\frac{1}{2}$ and the analysis is similar to Case 1. Similar argument applies in case $s_2,s_3 \in (\frac{1}{3}-\lambda,\frac{1}{3}+\lambda)$.
\item If $|\bl s \cap \left(\frac{1}{3}-\lambda,\frac{1}{3}+\lambda \right)|=1$: In case $s_1 \in (\frac{1}{3}-\lambda,\frac{1}{3}+\lambda)$, we have two sub-cases:
\begin{itemize}
\item if $s_2 = \frac{1}{3}+\lambda$, player 1 can improve her payoff by deviating to $\frac{1}{3}-\lambda$. 
\item if $s_2 > \frac{1}{3}+\lambda$, player 1 can improve her payoff by deviating to $\frac{1}{3}+\lambda$.
\end{itemize}
Else, if $s_3 \in (\frac{1}{3}-\lambda,\frac{1}{3}+\lambda)$ player 1 can improve her payoff by deviating to $\frac{1}{3}+\lambda$.

Finally, if $s_2 \in  (\frac{1}{3}-\lambda,\frac{1}{3}+\lambda)$ then 
$(s_1,s_3)$ must be $(\frac{1}{3}-\lambda,\frac{1}{3}+\lambda)$, since otherwise the corresponding player has a beneficial deviation. However, if $s_1=\frac{1}{3}-\lambda$, she can improve her payoff by deviating to $\frac{1}{3}+\lambda +\delta$ for small enough $\delta$.
\end{enumerate}
The analysis for $(\frac{2}{3}-\lambda,\frac{2}{3}+\lambda)$ is symmetric and hence omitted. We conclude that $\bl s$ is not a PNE.
\end{proof}

\begin{proof}[Proof of Proposition \ref{prop:clime n = 3}]
Assume w.l.o.g. that $s_1 \leq s_2 \leq s_3$. According to Claim \ref{claim:clime n=3}, if $\bl s$ is a PNE then the PIIs are free of facilities. As for the other alternatives,
\begin{enumerate}
\item $s_1=s_2=s_3$. Every player gets a payoff of $\frac{1}{3}$. However, by deviating to one of $\{ \frac{1}{2}-\delta,\frac{1}{2}+\delta \}$, a player can get more than that; hence, there exist beneficial deviations.
\item If $s_1<s_2<s_3$: In case $s_2 \not \in \{\frac{1}{3}+\lambda,\frac{2}{3}+\lambda \}$ then player 1 can deviate to $s_1'=s_2-\delta$ for small enough $\delta$ and improve her payoff. Else, if $s_2 =\frac{1}{3}+\lambda$, player 3 has  a beneficial deviation. Otherwise, if $s_2 =\frac{2}{3}+\lambda$, player 1 has a beneficial deviation.
\item Assume w.l.o.g.  $s_1=s_2<s_3$ (symmetrically for $s_1<s_2=s_3$): If $s_1 \not \in \{\frac{1}{3}-\lambda,\frac{2}{3}-\lambda\}$ then player 3 has a beneficial deviation to $s_3' = s_1+\delta$, for small enough $\delta$. If $s_1 = \frac{1}{3}-\lambda$, then unless $s_3=\frac{1}{3}+\lambda$ player 3 can improve her payoff by deviating to $\frac{1}{3}+\lambda$. However, if $s_3=\frac{1}{3}+\lambda$, player 1 can deviate to $\frac{1}{2}$ and improve her payoff. Similar argument holds for $s_1 = \frac{2}{3}-\lambda$.
\end{enumerate}
\end{proof}

\begin{proof}[Proof of Proposition \ref{prop: ic of clime}]
Denote
\[
I \defeq  \left(\frac{1}{n}-\frac{1}{\lambda},\frac{1}{n}+\frac{1}{\lambda}\right) \cup \left(\frac{n-1}{n}-\frac{1}{\lambda},\frac{n-1}{n}+\frac{1}{\lambda}\right).
\]
Observe that for every user $t\not\in I$ it holds that (See Line \ref{alg:clime as nime} in Algorithm \ref{alg:clime})
\[
\sum_{i=1}^n\clime(\bl s,t)_i \cdot\abs{s_i-t} =\min_{i' \in [n]} \abs{s_{i'}-t};
\]
thus,
\begin{align*}
\IC_n(\clime) &= \sup_{s\in S} \{ \SC(\clime,\bl s)-\SC(\nim,\bl s) \} \\
&= \sup_{s\in S} \left\{ \int_I  \left( \sum_{i=1}^n\clime(\bl s,t)_i \cdot\abs{s_i-t} -\min_{i' \in [n]} \abs{s_{i'}-t} \right)dt.  \right\}\\
& < \sup_{s\in S} \left\{ \int_I 1dt \right\} \\
&= 4\lambda.
\end{align*}
\end{proof}

}\fi}

\end{document}